\documentclass[a4paper,fullpage,10pt]{article}
\title{A hyperbolic problem with non-local
constraint describing ion-rearrangement in a model for ion-lithium batteries.}
\author{Stefano Scrobogna\footnote{Universit\'e de Bordeaux, France.},\hspace{2cm}
 Juan J.L.Vel\'azquez\footnote{IAM, Bonn, Germany.}}

\usepackage{fourier}
\usepackage[T1]{fontenc}

\usepackage[english]{babel}
\usepackage{amssymb,fullpage,amsthm}
\usepackage{mathtools}

\DeclareMathAlphabet{\mathcal}{OMS}{cmsy}{m}{n}

 \numberwithin{equation}{section}
 
\theoremstyle{theorem}
\newtheorem{theorem}{Theorem}[section]
\newtheorem{prop}[theorem]{Proposition}
\newtheorem{lemma}[theorem]{Lemma}
\newtheorem{cor}[theorem]{Corollary}

\theoremstyle{definition}
\newtheorem{definition}[theorem]{Definition}
\newtheorem{rem}[theorem]{Remark}

\usepackage{hyperref}
\def\res{\mathop{\mbox{Res}}}  
\def\d{\mathop{\mbox{d}}}
\def\xi{\mathop{\mbox{\textit{t}}}}

\begin{document}

 \maketitle
 
 \begin{abstract}
 In this paper we study the Fokker-Plank equation arising in a model which describes the charge and discharge process of ion-lithium batteries. In particular we focus our attention on slow reaction regimes with non-negligible entropic effects, which triggers the mass-splitting transition. At first we prove that the problem is globally well-posed. After that we prove a stability result under some hypothesis of improved regularity and a uniqueness result for the stability under some additional condition of the dynamical constraint driving the system.
 \end{abstract}

 \tableofcontents
 \setcounter{page}{1}
 \section{The model.}
 The model introduced in \cite{model}, to describe the charging and discharging of lithium-ion batteries, governs the evolution of a statistical ensemble of identical particles and is given by the non-local Fokker-Planck equation
 \begin{equation}\label{FP1}
 \tag{FP1}
 \tau\partial_t\rho(x,t)=\partial_x\Bigr(\nu^2\partial_x\rho(x,t)+\
 \left(H'(x)-\sigma(t)\right)\rho(x,t)\Bigr).
 \end{equation}
 Here $H$ is the free energy of a free particle  with thermodynamic state $x\in\mathbb{R}$. The probability density $\rho(\cdot,t)$ describes the state of the whole system at the time $t$, and $\sigma$ reflects that the system is subjected to some external forcing. Moreover, $\tau >0$ is the typical relaxation time of a single particle and $\nu >0$ accounts for entropic effect (stochastic fluctuation).\\
 The model \eqref{FP1} has two crucial features that cause highly non-trivial dynamics. First, the free energy $H$ is a double-well potential, hence there exist two different stable equilibria for each particle. Second, the system is not driven directly but via a time-dependent control parameter, in our case the parameter is the first moment of $\rho$, that means we impose the following dynamical constraint
 \begin{equation}
 \label{FP2}\tag{FP2}
 \int_\mathbb{R}x\rho(x,t)\text{d}x=\ell(t),
 \end{equation}
 where $\ell$ is some given function in time, and a direct calculation shows that \eqref{FP2} is equivalent to
 \begin{equation}
 \tag{FP3}\label{FP3}
 \sigma(t)=\int_\mathbb{R}H'(x)\rho(x,t)\text{d}x+\tau\dot{\ell}(t).
 \end{equation}
 The different relation between $\nu$ and $\tau$ may cause very different dynamical regimes, which have been studied in \cite{vel}. We are going to focus to what are called \textbf{slow reaction regimes}, in which we have the coupling
 \begin{equation}\label{coupling}
 \tau=\frac{a}{\log1/\nu},
 \end{equation}
 for some parameter $a\in \left(0,a_{\text{crit}}\right)$.\\
It has been seen in \cite{vel}
that under the assumption \eqref{coupling}, the solutions of \eqref{FP1}, \eqref{FP2} can be approximated, in the limit $\tau\to 0$ by means of some simpler problems. In particular,
during a suitable range of times, the solutions of \eqref{FP1}, \eqref{FP2} can be approximated by means of the solutions of the problem \eqref{mass splitting equation} -\eqref{ell} described later. It
turns out that during most of the times the function $\rho$ can be approximated by the sum of two Dirac masses. However, during the range of times in which the
approximation of \eqref{FP1}, \eqref{FP2} is valid, the mass of $\rho$ is distributed in a region
with size $x$ of order one. During those times the mass of $\rho$ is redistributed and,
in particular, the mass which is initially localized near the point $x_0$ is transported to two the neighborhood of two different points, denoted as $x_-, x_+$. 
This redistribution of the mass is described by the model \eqref{mass splitting equation} -\eqref{ell} and this is the issue considered in this paper. More details concerning the relationship
between the problems \eqref{FP1}, \eqref{FP2} and \eqref{mass splitting equation} -\eqref{ell} are given in \cite{vel}.\\
 
 The paper is divided as follows:
\begin{itemize}

\item
In Section \ref{intro} we give some simple assumption on the potential $H$ appearing in the equation \eqref{FP1} and we will introduce the problem that we study all along this paper. Moreover we give, in Definition \ref{condition H}, some assumptions that the potential $H$ has to satisfy in order that the problem makes sense and is well-posed. After that there is some technical result (namely Lemma \ref{dimostrazione condizione H}) which determinate some class of potentials which are admissible and compatible with the assumptions in Definition \ref{condition H}. \\

\item
In Section \ref{well posedness} we prove that
there is a unique (in a suitable space) solution of \eqref{mass splitting equation} -\eqref{ell} in the interval $\left(-\infty,\xi_0\right)$ with $\xi_0 >-\infty$ as explained heuristically in the paper \cite{vel}.\\

\item
In Section \ref{globalExistence} we prove that the problem is globally well posed in all $\mathbb{R}$, extending the local result proven in Section \ref{well posedness} performed in some interval of the form $\left[-\infty,\xi_0\right]$ to the whole real line.  \\

\item
In Section \ref{stability}, the last one, we finally prove that, up to sub-sequences (not relabeled) of re-scaled times $\left(\xi_m\right)_m$ such that $\xi_m\to\infty$ the problem converges to some equilibrium. At the end of such section, namely in Subsection \ref{uniqueness}, we prove as well that such equilibrium is unique and independent by the choice of the diverging sequence $\left(\xi_m\right)_m$ as long as the dynamical constraint $\ell(t)$ satisfies some condition in a vicinity of the critical time $\tilde{t}_2$ in which the mass splitting transition occurs.
\end{itemize}

\section{Introduction to the problem.}\label{intro}
 \subsection{Assumptions on the potential.}\label{assumptions on the potential}
In the paper we are going to assume the following hypothesis on the potential $H$
\begin{itemize}
\item[(A1)] $H$ is sufficiently smooth at least $\mathcal{C}^3_\text{loc}\left(\mathbb{R}\right)$, and
\begin{equation}\label{simple hypotesis on nonlinearity}
\begin{array}{r}
H'(x)= \alpha x +g(x),\\
\|H''\|_{L^\infty(\mathbb{R})}\leqslant c <\infty,
\end{array}
\end{equation}
with $g\in L^\infty$ and $\alpha >0$.
\item[(A2)] There exist constants $x_{\star\star}<x_\star<0<x^\star<x^{\star\star}$ and $\sigma_\star<0<\sigma^\star$ such that
\begin{enumerate}
\item $H'(x_\star)=H'(x^{\star\star})=\sigma^\star$.
\item$H'(x^\star)=H'(x_{\star\star})=\sigma_\star$.
\item for each $x\in\left(x_\star,x^\star\right)$ we have that $H'(x)\in\left(\sigma_\star,\sigma^\star\right)$,  $H''(x)<0$ and $H''(x)\geqslant 0$ for $x\in\left(x_\star,x^\star\right)^c$.\\
In particular the inverse of $H'$ has three strictly monotone branches
\begin{equation*}
\begin{array}{lll}
X_-&:\left(-\infty,\sigma^\star\right]&\to\left(-\infty,x_\star\right],\\
X_0&:\left[\sigma_\star,\sigma^\star\right]&\to\left[x_\star,x^\star\right],\\
X_+&:\left[\sigma_\star,\infty\right)&\to\left[\sigma_\star,\infty\right).
\end{array}
\end{equation*}
\end{enumerate}
\end{itemize}
In what follows we refer to $\left(x_\star,x^\star\right)^c$ as the stable interval, whereas the spinodal region $\left(x_\star,x^\star\right)$ is called the unstable interval. This nomenclature is motivated by the different properties of transport term in \eqref{FP1}. In both stable intervals adjacent characteristics approach each other exponentially fast, hence there is a strong tendency to concentrate mass into narrow peaks. In the unstable interval, however, the separation of adjacent characteristics de localizes at an exponential rate in time any peak with positive width.\\

\subsection{The mass splitting problem.}
\begin{rem}
This is a small remark about notations.\\ 
Suppose are given two functions
$$
F,G:U\subset\mathbb{R}^n\to\mathbb{R}.
$$
We write
$$
F\left(x_1,\ldots,x_n\right)\lesssim G\left(x_1,\ldots,x_n\right),
$$
if
$$
F\left(x_1,\ldots,x_n\right)\leqslant C\cdot G\left(x_1,\ldots,x_n\right),
$$
for some constant $C$ which is independent from the variables $\left(x_1,\ldots,x_n\right)\in U$.
\hspace*{\fill}$\blacklozenge$\medskip
\end{rem}

As explained in \cite{vel}, at the critical time $t_2\approx\tilde{t}_2$ we expect that the system undergoes a rapid transition from the unstable-stable configuration to a new stable-stable configuration. In order to describe this transition, in particular, to predict the mass distribution between the emerging stable peaks, we are going to study the \textbf{mass-splitting model}, which describes the peak widening effect (for a full description about the peak widening model we refer to \cite{vel}) on the rescaled time scale $s=(t-\tilde{t}_2)/\tau$ in the limit $\nu\to 0$. To make the notation simpler we will denote the re-scaled time $s$ simply $t$.  So, all in all the mass splitting problem consists of the following equations
\begin{align}
\label{mass splitting equation}\tag{MS1}
 \frac{\partial\rho}{\partial t}&=\frac{\partial}{\partial x}\left(\left(H'(x)-\sigma(t)\right)\rho\right),\\
 \tag{MS2} \label{rho with entropic effect}
 \rho(x,t)&\sim\frac{m}{\sqrt{\pi}e^{|H''(x_0)|t}}\exp\left(-\frac{(x-x_0)^2}{e^{2|H''(x_0)|t}}\right)+(1-m)\delta_{x_s}.
\end{align}
The asymptotics in \eqref{rho with entropic effect} takes place as long as $t\to -\infty$. The meaning of this asymptotic formula will be explained later, in Definition \ref{condition H}.\\
\eqref{rho with entropic effect} codify the necessity
 to impose asymptotic initial conditions at $\xi=-\infty$, which have to reflect the fact that the mass-splitting process starts in an unstable-stable configuration of a two-peaks model, and that the peak on the left is a rescaled Gaussian due to entropic randomness, which means as $\xi\to -\infty,s\in\left\lbrace+,-\right\rbrace$, although in the following we are going to assume $s=+$, the case $s=-$ is simply symmetric, and
 \begin{equation}\tag{MS3}
 \label{ell}
 \int x\rho \d x=\ell^\star\in\mathbb{R},
 \end{equation}
 this because the transition happens in a timespan so short that allows us to consider the dynamical constraints $\ell(t)\approx\ell^\star$ during the whole transition process.\\
The mass-splitting problem hence consists of equations \eqref{mass splitting equation},\eqref{rho with entropic effect} and \eqref{ell}.\\ 
  Moreover, as well as before we require to have well prepared initial data, i.e.
 \begin{eqnarray}
a=H''(x_0)<0\label{a},\\
b=H''(x_+)>0\label{b},\\
H'(x_0)=H'(x_+)=\sigma_0\in (\sigma_\star,\sigma^\star)\label{c},\\
m\in (0,1],\\
mx_0+(1-m)x_+=\ell^\ast,\label{d}
\end{eqnarray}
where indeed \eqref{a}--\eqref{c} codify the unstable-stable configuration in which the process starts, and \eqref{d} is a compatibility condition of the initial data with the dynamical constraint.\\ 

In the following it will turn out to be more convenient to use the distribution of $\rho$ instead of the density itself, to this end we define:
\begin{equation}
R(x,t)=\int_{-\infty}^x \rho (y,t)\d y. \label{distribuzione}
\end{equation}
Combining \eqref{mass splitting equation} and \eqref{distribuzione} we obtain the following:
\begin{equation}
\label{mass splitting equation 2}
\frac{\partial R}{\partial t}=\left(H'(x)-\sigma\right)\frac{\partial R}{\partial x}.
\end{equation}

Moreover notice that multiplying \eqref{mass splitting equation} by $x$ and integrating in the real line, we obtain, after some integration by parts and using \eqref{ell} and \eqref{distribuzione}:
\begin{equation} \label{definizione sigma}
\sigma(t)=\int H'(x)\frac{\partial R}{\partial x}\d x.
\end{equation}

 Integrating \eqref{rho with entropic effect} in the interval $(-\infty,x)$ we obtain the formal asymptotic:
 
\begin{equation}
\label{asymptotic R}
R(x,t)\sim m Q\left(\frac{x-x_0}{e^{-at}}\right)+(1-m)\chi_{\{x\geqslant x_+\}}(x),
\end{equation}
as $t\to -\infty$. This convergence takes place in $L^1$, although it is uniform in compact sets of $\mathbb{R}\backslash \{x_+\}$. In particular, considering a (small) neibourhood $\mathcal{U}_{x_0}$ of $x_0$, $R\left(x_0 + ye^at, t\right)\to mQ\left(y\right)$ as $t\to -\infty$ in $\mathcal{C}^0_{\text{loc}}\left(\mathcal{U}_{x_0}\right)$.
 The function $Q$ is given by
$$
Q(y)=\frac{1}{\sqrt{\pi}}\int_{-\infty}^y e^{-\eta^2}\d\eta.
$$

The aim of this paper is to show, under which conditions the transition of the system gives, at the rescaled time $t=+\infty$, a new stable-stable configuration, i.e. that there exist two non-spinodal states $\hat{x}_-,\hat{x}_+$ and a numerical value $\tilde{m}\in\left[-1,1\right]$ such that
\begin{align*}
H'(\hat{x}_-)&=H'(\hat{x}_+),&(m_--\tilde{m})\hat{x}_-+(m_++\tilde{m}
)\hat{x}_+&=\ell^\star=m_-x_-+m_+x_+.
\end{align*}

The existence and the uniqueness of $\tilde{m}$ is not obvious since the mass-splitting problem involves two subtle limits.  First one has to show that the asymptotic condition \eqref{rho with entropic effect} gives rise to a well-posed initial value problem at $t=-\infty$. Second, one has to guarantee that solutions do not drift as $t\to\infty$ along the connected one-parameter family of equilibrium solutions.\\
We need some general assumption on the function $H$ in order to prove that the problem is well-posed.

\begin{definition}\label{condition H}
We say that a function $H\in \mathcal{C}^2_{\text{loc}}(\mathbb{R})$ satisfies the \textbf{condition H} if, for any function $\phi(t)=\sigma(t)-\sigma_0\in \mathcal{C}^\infty(-\infty,T),T\in\mathbb{R}$ satisfying $|\phi(\xi)|\leqslant Me^{-\left( 2a+\delta  \right) t}, \delta >0$ and any $K\in\mathbb{R}$, there exist a unique solution to the ODE problem
\begin{equation} \label{equation X}
\left\lbrace
\begin{array}{lcr}
\frac{\d}{\d t}X(t,K)=-\left( H'(X(t,K)) -\sigma(t)\right),\\[6mm]
X(t,K)=x_0+Ke^{-at}+Y\left(K,t\right),&\text{as}&t\to -\infty.
\end{array}
\right.
\end{equation}
With $Y(K,t)=o\left(e^{(-a+\delta)t}\right)$.\\
Moreover for every fixed $t\in (-\infty,T]$ the transformation
$$
K\in\mathbb{R}\mapsto X(t,K),
$$
is a one-to-one transformation of the real line in a interval $\left(X_-(t),X_+(t)\right)$ where the functions $X_\pm(t)$ solves the  problem
\begin{equation}
\label{char equtaions}
\left\lbrace
\begin{array}{lcr}
\frac{\d}{\d t}X_\pm(t)=-\left( H'(X_\pm(t)) -\sigma(t)\right),\\[6mm]
X_\pm(t)=x_\pm +Y_\pm(t) , &\text{as}&t\to -\infty,
\end{array}
\right.
\end{equation}
With $Y_\pm(t)=o\left(e^{(-a+\delta)t}\right)$.\\
Where $H'(x_0)=H'(x_\pm)=\sigma_0$. Finally for any fixed $t\in (-\infty,T]$ we have
$$
\lim_{K\to \pm\infty}X(t,K)=X_\pm (t).
$$
\end{definition}
Definition \ref{condition H} is vacuous if the set of potentials satisfying such condition is empty. For this reason the following lemma gives an example of a class of potentials that satisfy the Condition H.
\begin{lemma}\label{dimostrazione condizione H}
Every function $H\in\mathcal{C}^3_{\text{loc}}(\mathbb{R})$ such that satisfies the compatibility condition
$
H'(x_-)=H'(x_0)=H'(x_+)=\sigma_0
$
 satisfy the conditions in Definition \ref{condition H}.
\end{lemma}
\begin{proof}
We are going to subdivide the proof in several steps
\begin{itemize}
\item[\textbf{Step 1}]\label{step1} We want to show that the solutions of \eqref{char equtaions} are well defined. We are going to do this for the characteristic $X_+$, the other case is similar. Let us write $X_+$ in the following form
$$
X_+(t)=x_++Y_+(t),
$$
where $Y_+$ is considered to be a perturbation. Than the equation \eqref{char equtaions} reads as
$$
Y_+'=-bY_++\phi(t)-\mathcal{O}\left(Y_+^2\right),
$$
and hence if such a solution exists than it has to take the following form
$$
Y_+(t)=\int_{-\infty}^t e^{-b(t-s)}\left[\phi(s)-\mathcal{O}\left(Y_+^2(s)\right)\right]\d s.
$$
Define the following operator
\begin{equation}
\label{T^+}
T^+[Y_+](t)=\int_{-\infty}^t e^{-b(t-s)}\left[\phi(s)-\mathcal{O}\left(Y_+^2(s)\right)\right]\d s,
\end{equation}
we show that $T^+$ admits a fixed point in the space 
$$
J(t_0,\delta)=\left\lbrace f:|f(t)|\leqslant e^{(-a+\delta)t},t<t_0,\delta >0\right\rbrace,
$$
endowed with the norm
$$
\|f\|_{J(t_0,\delta)}=\sup_{t<t_0}\left\lbrace\left|f(t)\right|\cdot e^{(a-\delta)t}\right\rbrace,
$$
for $t_0$ sufficiently negative and $\delta$ small.\\
To do so take $Y_+\in J(t_0,\delta)$ and evaluate
$$
|T^+[Y_+](t)|=\left|\int_{-\infty}^t e^{-b(t-s)}\left[\phi(s)-\mathcal{O}\left(Y_+^2(s)\right)\right]\d s\right|\\
\leqslant c_1 e^{-(2a+\delta)t)}+c_2 e^{2(-a+\delta)t}\leqslant e^{(-a+\delta)t},
$$
for $\xi$ sufficiently negative, since $c_i=c_i(M,a,b,\delta)$.\\
At this point take $Y_{+,1},Y_{+,2}\in J(t_0,\delta)$, and consider
\begin{equation}
\label{CoNtrazione}
\left|T^+[Y_{+,1}](t)-T^+[Y_{+,2}](t)\right|\leqslant\int_{-\infty}^t e^{-b(t-s)}\left|\mathcal{O}\left(Y_{+,1}^2(s)\right)-\mathcal{O}\left(Y_{+,2}^2(s)\right)\right|\d s.
\end{equation}
In particular we point out that the functions $\mathcal{O}\left(Y_{+,i}^2(s)\right)$ have an obvious explicit expression, as explained in equation \eqref{rho_+}, in particular, in our case
$$
\mathcal{O}\left(Y_{+,i}^2(s)\right)=\rho_{+,i}(t)=H'(x_+ +Y_{+,i}(t))-\left(H'(x_+)+bY_{+,i}(t)\right).
$$
With this consideration we can rewrite equation \eqref{CoNtrazione} as
\begin{multline}\label{COntrazione}
\left|T^+[Y_{+,1}](t)-T^+[Y_{+,2}](t)\right|\leqslant \\
\int_{-\infty}^t e^{-b(t-s)}\left(\left|H'(x_+ +Y_{+,1}(s))-H'(x_+ +Y_{+,2}(s))\right|+b\left|Y_{+,1}(s)-Y_{+,2}(s)\right|\right)\d s.
\end{multline}
In particular, since $H\in\mathcal{C}^3_{\text{loc}}(\mathbb{R})$ we can say that $H'\in\mathcal{C}^{0,1}_{\text{loc}}(\mathbb{R})$, which means that
$$
\left|H'(x_+ +Y_{+,1}(s))-H'(x_+ +Y_{+,2}(s))\right|\leqslant L \left|Y_{+,1}(s)-Y_{+,2}(s)\right|,
$$
for some $L>0$ that depends only on $H$ and on a compact $\mathcal{K}\subset\mathbb{R}$ sufficiently large.\\
With this consideration we can rewrite \eqref{COntrazione} as
\begin{multline*}
\left|T^+[Y_{+,1}](\xi)-T^+[Y_{+,2}](\xi)\right|\leqslant 
C\left(\mathcal{K},H,b\right)\int_{-\infty}^{\xi} e^{-b(\xi-s)}\left|Y_{+,1}(s)-Y_{+,2}(s)\right| \d s
\\ \lesssim
\|Y_{0,1}-Y_{0,2}\|_{J(\xi_0,\delta)}\cdot e^{-b\xi}\int_{-\infty}^{\xi} e^{(b-a+\delta)s}\d s\lesssim 
e^{(-a+\delta)\xi}\cdot \|Y_{0,1}-Y_{0,2}\|_{J(\xi_0,\delta)},
`\end{multline*}
proving that it is a contraction for $\xi\leqslant\xi_0$ and concluding the proof of the first step.
\item[\textbf{Step 2}] We want to prove that for any fixed $K$ there exists a $t_\star=t_\star (K)$ such that \eqref{equation X} has a unique solution.\\
To prove this we write
$$
X(K,\xi)=x_0+Ke^{-a\xi}+Y(K,\xi),
$$
the proof is the same as in the first step only considering $Y(K,\xi)$ belonging to the following space
$$
J(t_{0},\delta,K)=\left\lbrace f\in\mathcal{C}^0|f(\xi)|\leqslant C_K e^{(-a+\delta)\xi},\xi<t_{0},\delta >0\right\rbrace,
$$
with $\rho$ defined as in \eqref{def rho}, with the norm
$$
\|f\|_{J(\xi_0,\delta,K)}=\sup_{\xi<\xi_0}\left\lbrace\left|f(\xi)\right|\cdot e^{(a-\delta)\xi}\right\rbrace.
$$
We point out the fact that, a priori, could happen that 
$
\lim_{K\to\pm\infty}t_\star (K)=-\infty.
$
\item[\textbf{Step 3}] In this step we prove that there exist a time $\bar{\xi}=\bar{\xi}(\phi,\delta)$ such that the solutions of \eqref{equation X} are defined uniformly for all $\xi\leqslant\bar{\xi}$.\\
To do so fix a $K$ that we might as well consider to be positive. Since as long as $|X'(\tilde{\xi},K)|<\infty$ we can extend the solutions of \eqref{equation X} to the right of $ \tilde{\xi}$, see \cite[Theorem 1.1, Chapter 2]{CL} for such extensibility theorem, we suppose than that there exists a $\hat{\xi}=\hat{\xi}(K)$ such that
\begin{equation}
\label{infinito!}
X'(\xi,K)\xrightarrow{\xi\to\hat{\xi}}\infty.
\end{equation}
For instance $X'(\xi,K)\xrightarrow{\xi\to\hat{\xi}}+\infty$. The case $X'(\xi,K)\xrightarrow{\xi\to\hat{\xi}}-\infty$ is similar.\\
If this happens, than, since $X'=\sigma_0+\phi(\xi)-H'(X)\to +\infty$ then it must be true that
$$
-H'\left(X'(\xi,K)\right)\xrightarrow{\xi\to\hat{\xi}}+\infty,
$$
and hence considering the structure of $H'$
$$
X'(\xi,K)\xrightarrow{\xi\to\hat{\xi}}-\infty,
$$
which contradicts the hypothesis \eqref{infinito!}.
\item[\textbf{Step 4}] What is left to prove in this last step is that
$$
\lim_{K\to \pm\infty}X(\xi,K)=X_\pm (\xi).
$$
First we claim that for every $\xi\leqslant\bar{\xi}$ there exist a $K>0$ such that $X'(K,\xi)>0$. As usual for $K<0$ the procedure is similar.\\
Indeed as we have seen in the second step for every rescaled time $\xi$ there exists a $\kappa$ small enough such that the following approximation holds
\begin{equation}
\label{approssimazione caratt}
X(\kappa,\xi)=x_0+\kappa e^{-a\xi}+o\left(e^{-a\xi}\right).
\end{equation}
But then
$$
X'(\kappa,\xi)=\sigma(\xi)-H'\left(X(\kappa,\xi)\right)=-a\kappa e^{-a\xi}+o\left(e^{-a\xi}\right)>0.
$$
Where $a$ is defined in \eqref{a}. Since the approximation in equation \eqref{approssimazione caratt} is valid for every $\kappa$ as $\xi\to -\infty$ we can say that
$
\lim_{\xi\to-\infty}\kappa(\xi)=+\infty.
$
Now, considering the flux structure of \eqref{equation X}, we can argue that $X'(K,\xi_+)>X'(K,\xi)$ in some right neighbourhood of $\xi$. I want to show that, given an $\eta>0$ that we are going to consider small, such that $X'(K,\xi)>\eta$, the system reaches in a finite time the state
\begin{equation}
\label{stato}
X'\left(K,\hat{\xi}\right)=\eta,
\end{equation}
where it is not restrictive to set
$$
\hat{\xi}=\inf\left\lbrace {\xi}_+>\xi:X'\right(K,{\xi}_+\left)=\eta\right\rbrace.
$$
We remark the fact that, for every $\eta >0$ and small there exists a state satisfying \eqref{stato}, this because if we suppose $X'\left(K,\hat{\xi}\right)\geqslant\bar{\eta}>0$ for every $\hat{\xi}\geqslant \xi$ we would obtain that
$$
X\left(K,\hat{\xi}\right)-X\left(K,\xi\right)\geqslant \bar{\eta}\left(\hat{\xi}-\xi\right),
$$
hence, since $X_+(\xi)$ is bounded (see Step 1) there would exist a finite time $\tilde{\xi}$ such that
$
X\left(K,\tilde{\xi}\right)=X_+\left(\tilde{\xi}\right),
$
but this contradicts the uniqueness of solution for \eqref{equation X}.

We want to see that under these assumptions $X(K,s)$ reaches the configuration $X\left(K,\hat{\xi}\right)$ in a finite time.\\
Since the characteristics are uniformly bounded as we have seen in the third step we can say that
$
X\left(K,\hat{\xi}\right)-X\left(K,\xi\right)\leqslant C_0,
$
independently of $\xi,\hat{\xi}$. Moreover 
$$
X\left(K,\hat{\xi}\right)-X(K,\xi)=\int_{\xi}^{\hat{\xi}}X'(K,z)\d z=\left(\hat{\xi}-\xi\right)X'(K,\bar{t}).
$$
Where in the last equality we have used the mean value theorem. Moreover
$$
X'\left(K,\bar{t}\right)\geqslant\inf_{\zeta\in\left[\xi,\hat{\xi}\right]}X'(K,\zeta)>\frac{\eta}{N}>0,
$$
since $X'(K,s)>0$ in $\left[\xi,\hat{\xi}\right]$. Whence
$
(\hat{\xi}-\xi)\leqslant N\frac{C_0}{\eta}<\infty.
$
Once we have proved that such a state can be reached, in a possibly large, but finite time, and since all this procedure has been done independently by the starting time $\xi$ we can state that, taking $\xi$ sufficiently negative there exists a $\hat{\xi}$ such that
\begin{equation}\label{ausiliaria}
H'\left(X\left(K,\hat{\xi}\right)\right)=\sigma_0+\phi\left(\hat{\xi}\right)-\eta.
\end{equation}
But $\phi\left(\hat{\xi}\right)$ is an $o\left( e^{-(2a+\delta)\hat{\xi}}\right)$ function, hence we can assert that for $\hat{\xi}$ sufficiently negative
$$
\sigma_0-\frac{3}{2}\eta\leqslant
\sigma_0+\phi\left(\hat{\xi}\right)-\eta
\leqslant \sigma_0-\frac{\eta}{2},
$$
whence
\begin{equation}\label{ausiliaria2}
\sigma_0+\phi\left(\hat{\xi}\right)-\eta=\sigma_0+\mathcal{O}\left(\eta\right).
\end{equation}
At this point, since $X\left(K,\hat{\xi}\right)>X\left(K,\xi\right)$, and since the state $X\left(K,\hat{\xi}\right)$ has to satisfy \eqref{ausiliaria}, considering moreover \eqref{ausiliaria2} we can say that
$
X\left(K,\hat{\xi}\right)=x_++\mathcal{O}\left(\eta\right),
$
and, thanks to Step 1 we can assert
$
X_+\left(\hat{\xi}\right)=x_++o(1),
$
whence
\begin{equation}
\label{ausiliaria3}
\left|X\left(K,\hat{\xi}\right)-X_+\left(\hat{\xi}\right)\right|\leqslant \mathcal{O}\left(\eta\right)+o(1),
\end{equation}
and indeed the right hand side of \eqref{ausiliaria3} can be made as small as we want, concluding the proof.
\end{itemize}
\end{proof}

\section{Well posedness for $\xi\to -\infty$.}\label{well posedness}

 We want to see that there effectively exist some $\rho$ satisfying \eqref{mass splitting equation}, \eqref{rho with entropic effect}, formalized we want to prove the following proposition.
 \begin{prop}\label{well posedness at -infty}
 Suppose $H$ satisfies the conditions in Definition \ref{condition H}, and $H'''$ absolutely continuous in compact sets of $\mathbb{R}$. Suppose also that the following stability condition holds:
 \begin{equation}\label{stability condition for well posedness}
 (1-m)a+mb>0
 \end{equation}
 with $a$ and $b$ defined respectively as in \eqref{a} and \eqref{b}. For any $m\in (0,1]$ there exist a $T\in\mathbb{R}$ and a unique solution $R$ of \eqref{mass splitting equation 2}, \eqref{definizione sigma}, monotonically increasing in $x$ for any $t\in (-\infty,T]$ such that $R(x,t)\in [0,1]$ and satisfying \eqref{asymptotic R}, i.e.
 \begin{equation}
 \label{asymptotic R in L1}
 \left\|R(\bullet,\xi)-m\chi_{\{x\geqslant x_0\}}-(1-m)\chi_{\{x\geqslant x_+\}}\right\|_{L^1(\mathbb{R})}\to 0,
 \end{equation}
 as $\xi\to -\infty$.
 
 \end{prop}
 \begin{proof} In order to prove the existence of solutions we define the following class of functions:
 $$
 K_{M,T}(\delta)=\left\lbrace
 \begin{array}{lcr}
 \sigma\in\mathcal{C}^0|\sigma(\xi)-\sigma_0|\leqslant M e^{-(2a+\delta)\xi}&,&\xi\in (-\infty,T]
 \end{array}
  \right\rbrace ,
 $$
for some $\delta >0$ and sufficiently small.\\ 
 
 Our goal is to show that for $T<0,|T|,M$ sufficiently large it is possible to define a transformation from $ K_{M,T}(\delta)$ to itself, whose fixed point is equivalent to solving \eqref{mass splitting equation 2}, \eqref{definizione sigma} and \eqref{asymptotic R in L1}.\\
 Given $\sigma\in K_{M,T}(\delta)$ and $K\in\mathbb{R}$ let $X(s,K,\sigma)=X(s,K)$ and $X_\pm(s,\sigma)=X_\pm(s)$ respectively the solutions of \eqref{equation X} and \eqref{char equtaions}, which are well defined due to the fact that $H$ satisfies the conditions in Definition \ref{condition H}. We define then $R(x,s,\sigma)$ as follows:
 \begin{equation}
 \label{R per contrazione}
 \left\lbrace
 \begin{array}{lcr}
 R(X(\xi,K,\sigma),\xi)=m Q(K)+(1-m)\chi_{\left\lbrace x\geqslant X_+(\xi,\sigma) \right\rbrace}\left(X(\xi,K,\sigma)\right),\\
 R(x,\xi,\sigma)=0,&\text{if}&x<X_-(\xi,\sigma),\\
 R(x,\xi,\sigma)=1,&\text{if}&x>X_+(\xi,\sigma).
 
 \end{array}
 \right.
 \end{equation}
 
 Our goal is to obtain a function $\sigma\in K_{M,T}(\delta)$ such that satisfies equation \eqref{definizione sigma}, with $R$ defined by means of \eqref{R per contrazione}. We want to obtain a linearised version of \eqref{definizione sigma} where it will be implicitly assumed that$\frac{\partial R}{\partial x}$ can be approximated for $\xi\to -\infty$ as two dirac masses at $X_0(\xi,\sigma)=X(\xi,0,\sigma)$ and $X_+(\xi,\sigma)$ respectively. To this end rewrite \eqref{definizione sigma} as:
 
$$ 
\sigma(\xi)=H'(X_0(\xi,\sigma))\int^{X_+(\xi,\sigma)}_{X_-(\xi,\sigma)}\frac{\partial R}{\partial x}(x,\xi,\sigma)\d x+\\
 \int^{X_+(\xi,\sigma)}_{X_-(\xi,\sigma)}\left[
 H'(x)-H'(X_0(\xi,\sigma))\right]\frac{\partial R}{\partial x}(x,\xi,\sigma)\d x+
 (1-m)H'(X_+(\xi,\sigma)),
$$
 whence we obtain
 \begin{equation}\label{whatever1}
 \sigma(\xi)=m H'(X_0(\xi,\sigma))+(1-m)H'(X_+(\xi,\sigma))+
 \int^{X_+(\xi,\sigma)}_{X_-(\xi,\sigma)}\left[
 H'(x)-H'(X_0(\xi,\sigma))\right]\frac{\partial R}{\partial x}(x,\xi,\sigma)\d x.
 \end{equation}
 
 At this point, using Taylor expansion we may write
 \begin{equation}
 \label{whatever2}
 \begin{array}{l}
 H'(X_0(\xi,\sigma))=H'(x_0)+H''(x_0)\left(X_0(\xi,\sigma)-x_0\right)+\rho_0(\xi,\sigma),\\
 H'(X_+(\xi,\sigma))=H'(x_+)+H''(x_+)\left(X_+(\xi,\sigma)-x_+\right)+\rho_+(\xi,\sigma),
 \end{array}
  \end{equation}
  where indeed $\rho_0,\rho_+$ are reminder with the obvious explicit expression:
  
 \begin{align} 
 \rho_0(\xi,\sigma)&= H'(X_0(\xi,\sigma))-\left(H'(x_0)+H''(x_0)\left(X_0(\xi,\sigma)-x_0\right) \right),\label{rho_0}\\
 \rho_+(\xi,\sigma)&=H'(X_+(\xi,\sigma))-\left( H'(x_+)+H''(x_+)\left(X_+(\xi,\sigma)-x_+\right)\right),\label{rho_+}
 \end{align}

 hence, considering \eqref{whatever2} into \eqref{whatever1} we obtain the following expression for $\sigma(\xi)$:
 \begin{multline*}
 \sigma(\xi)=\left[m H'(x_0)+(1-m)H'(x_+)\right]+ma\left(X_0(\xi,\sigma)-x_0\right)+(1-m)b\left(X_+(\xi,\sigma)-x_+\right)+\\
 \int^{X_+(\xi,\sigma)}_{X_-(\xi,\sigma)}\left[
 H'(x)-H'(X_0(\xi,\sigma))\right]\frac{\partial R}{\partial x}(x,\xi,\sigma)\d x+m \rho_0(\xi,\sigma)+(1-m) \rho_+(\xi,\sigma),
 \end{multline*}
 with $a,b$ as always defined by means of \eqref{a}, \eqref{b}.\\
 Notice at this point that $ \left[m H'(x_0)+(1-m)H'(x_+)\right]=\sigma_0 $, set $\phi=\sigma-\sigma_0$ and set
 \begin{equation}
 \label{gei}
 J(\xi,\sigma)=\int^{X_+(\xi,\sigma)}_{X_-(\xi,\sigma)}\left[
 H'(x)-H'(X_0(\xi,\sigma))\right]\frac{\partial R}{\partial x}(x,\xi,\sigma)\d x,
 \end{equation}
with this consideration we can write equation \eqref{whatever1} as
\begin{equation}
\label{whatever3}
\phi(\xi)=ma\left(X_0(\xi,\sigma)-x_0\right)+(1-m)b\left(X_+(\xi,\sigma)-x_+\right)\\
+m \rho_0(\xi,\sigma)+(1-m) \rho_+(\xi,\sigma)+J(\xi,\sigma).
\end{equation}

Our next aim is to express $J(\xi,\sigma)$ in a more suitable form,  to do so  integrate by parts equation \eqref{gei}  obtaining that
$$
J(\xi,\sigma)=-\int_{X_-(\xi,\sigma)}^{X_+(\xi,\sigma)}H''(x)R(x,\xi,\sigma)\text{d}x
+R\left(X_+(\xi,\sigma),\xi,\sigma\right)\Bigr[H'(X_+(\xi,\sigma)-H'(X_0(\xi,\sigma)\Bigr].
$$
Moreover, using the following change of variable
\begin{align*}
x&=X(s,K,\sigma),\\
\text{d}x&=\frac{\partial X(s,K,\sigma)}{\partial K}\text{d}K,
\end{align*}
and considering the Condition (H), that ensures that the interval $K\in\mathbb{R}$ is transformed by $X(s,\cdot,\sigma)$ into $\left(X_-(s,\sigma),X_+(s,\sigma)\right)$ we obtain, using moreover that $R\left(X_+(\xi,\sigma),\xi,\sigma\right)=m_0$
$$
J(\xi,\sigma)=-\int_\mathbb{R}H''\left(X(\xi,K,\sigma)\right)R\left(X(\xi,K,\sigma),\xi,\sigma\right)\frac{\partial X(s,K,\sigma)}{\partial K}\text{d}K+\\
+m_0\Bigr[ H'\left(X_+(\xi,\sigma)\right)-H'\left(X_0(\xi,\sigma)\right)\Bigr].
$$
and, using also \eqref{R per contrazione}
\begin{align*}
J(\xi,\sigma)&
=-m_0\int_\mathbb{R}H''\left(X(\xi,K,\sigma)\right)Q(K)\frac{\partial X(s,K,\sigma)}{\partial K}\text{d}K+m_0\Bigr[ H'\left(X_+(\xi,\sigma)\right)-H'\left(X_0(\xi,\sigma)\right)\Bigr]\\
&=-m_0\int_\mathbb{R}\frac{\partial}{\partial K}\bigr(H'\left(X(\xi,K,\sigma)\bigr)\right)Q(K)\text{d}K+m_0\Bigr[ H'\left(X_+(\xi,\sigma)\right)-H'\left(X_0(\xi,\sigma)\right)\Bigr]\\
&=m_0\int_\mathbb{R}\left[ H'\left(X(\xi,K,\sigma)\right)-H'\left(X_+(\xi,\sigma)\right)\right]Q'(K)\text{d}K\\
&+m_0\Bigr[ H'\left(X_+(\xi,\sigma)\right)-H'\left(X_0(\xi,\sigma)\right)\Bigr].
\end{align*}
Obtaining hence, in the end that
\begin{equation}\label{whatever4}
J(\xi,\sigma)=m\int_\mathbb{R}\left(H'(X(\xi, K,\sigma))-H'(X_0(\xi,\sigma))\right)Q'(K)\d K.
\end{equation}
At this point we want to linearise equation \eqref{equation X} and \eqref{char equtaions} using Taylor expansion for $H'$ around $x_0$ and $x_+$, obtaining the following equations
\begin{align}
\frac{\d}{\d\xi}(X_0(\xi,\sigma)-x_0)&=-a (X_0(\xi,\sigma)-x_0)+\phi(\xi)-\rho_0(\xi,\sigma),\label{linearizzazione X0}\\
\frac{\d}{\d\xi}(X_+(\xi,\sigma)-x_+)&=-b (X_+(\xi,\sigma)-x_+)+\phi(\xi)-\rho_+(\xi,\sigma),\label{linearizzazione X+}\\
\frac{\d}{\d\xi}X(\xi,K,\sigma)&=-a \left( X(\xi,K,\sigma)-x_0\right)+\phi(\xi)-\rho(\xi,K,\sigma),\label{linearizzazione X}
\end{align}
with $\rho$ in \eqref{linearizzazione X} reminder of the form
\begin{equation}
\label{def rho}
\rho(\xi,K,\sigma)=\left[H'(X(\xi,K,\sigma) )-H'(x_0)-a(X(\xi,K,\sigma)-x_0)  \right],
\end{equation}

and not a probability density.\\

We want to linearise also equation \eqref{whatever4}, to do so write first
$$
\begin{array}{rcl}
\left(H'(X(\xi, K,\sigma))-H'(X_0(\xi,\sigma))\right)&=&\left(H'(X(\xi, K,\sigma))-\sigma_0\right)-\left(H'(X_0(\xi,\sigma))-\sigma_0\right),\\[3mm]
&=&\left(H'(X(\xi, K,\sigma))-H'(x_0)\right)-\left(H'(X_0(\xi,\sigma))-H'(x_0)\right),

\end{array}
$$
whence we obtain
\begin{multline}\tag{a}\label{lin1}
J(\xi,\sigma)=m a \int_\mathbb{R}(X(\xi,K,\sigma)-x_0)Q'(K)\d K-m b \int_\mathbb{R}(X_+(\xi,\sigma)-x_+)Q'(K)\d K+\\
+m \int_\mathbb{R} \rho(\xi,K,\sigma)Q'(K)\d K-m \int_\mathbb{R}\rho_+(\xi,\sigma)Q'(K)\d K,
\end{multline}
and also, considering that $\int Q'(K)\d K=1$,
\begin{equation}
\label{lin2}\tag{b}
J(\xi,\sigma)=m a \int_\mathbb{R}(X(\xi,K,\sigma)-x_0)Q'(K) \d K- m a (X_0(\xi,\sigma)-x_0)+\\
m \int_\mathbb{R} \rho(\xi,K,\sigma)Q'(K)\d K - m \rho_0(\xi,\sigma).
\end{equation}
We remark the fact that \eqref{lin1} and \eqref{lin2} are different linearisation of $J(\xi,\sigma)$.\\

On the other hand we can remove the leading order in the asymptotic of $ X(\xi,K,\sigma)-x_0 $ in order to rewrite the problem as a linearised one plus a perturbative term. We write:
\begin{equation} \label{lin problem plus perturbative term}
\begin{array}{rcl}
X(\xi,K,\sigma)-x_0&=& Ke^{-a \xi}+Y(\xi,K,\sigma),\\
X_0(\xi,\sigma)-x_0&=&Y_0(\xi,\sigma),\\
X_+(\xi,\sigma)-x_+&=&Y_+(\xi,\sigma).
\end{array}
\end{equation}
Substituting in \eqref{whatever3} the results in \eqref{lin problem plus perturbative term} we obtain
$$
\phi(\xi)=m\left[H''(x_0)Y_0(\xi,\sigma)+\rho_0(\xi,\sigma)\right]+(1-m)\left[H''(x_+)Y_+(\xi,\sigma)+\rho_+(\xi,\sigma)\right]+J(\xi,\sigma).
$$
Considering moreover the result in equation \eqref{lin2}
\begin{multline*}
\phi(\xi)=m\left[H''(x_0)Y_0(\xi,\sigma)+\rho_0(\xi,\sigma)\right]+(1-m)\left[H''(x_+)Y_+(\xi,\sigma)+\rho_+(\xi,\sigma)\right]+\\
m a \int_\mathbb{R}Y(\xi,K,\sigma)Q'(K)\d K- m a Y_0(\xi,\sigma)+
m \int_\mathbb{R} \rho(\xi,K,\sigma)Q'(K)\d K - m \rho_0(\xi,\sigma).
\end{multline*}

So, at the end, considering that $\int Q'\text{d}K=1$ we obtain
\begin{multline}\label{phi qualcosa}
\phi(\xi)=m\left( a Y_0(\xi,\sigma )+\rho_0(\xi,\sigma  )\right)+(1-m)\left( b Y_+(\xi,\sigma )+\rho_+(\xi,\sigma ) \right) +\\
+m a \int_\mathbb{R} (Y(\xi,K,\sigma)-Y_0(\xi,\sigma))Q'(K)\d K+m \int_\mathbb{R} (\rho (\xi,K,\sigma)-\rho_0(\xi,\sigma))Q'(K)\d K,
\end{multline}
and indeed considering \eqref{lin problem plus perturbative term} and linearising  using Taylor expansion we obtain:
\begin{align}
\frac{\d}{\d\xi}Y_0(\xi,\sigma)&=-a Y_0(\xi,\sigma)+\phi(\xi)-\rho_0(\xi,\sigma),\label{evolution Y0}\\
\frac{\d}{\d\xi}Y_+(\xi,\sigma)&=-b Y_+(\xi,\sigma)+\phi(\xi)-\rho_+(\xi,\sigma),\label{evolution Y+}\\
\frac{\d}{\d\xi}Y(\xi,K,\sigma)&=-a Y(\xi,K,\sigma)+\phi(\xi)-\rho(\xi,K,\sigma).\label{evolution Y}
\end{align}
 \eqref{evolution Y0}--\eqref{evolution Y} are simply \eqref{linearizzazione X0}--\eqref{linearizzazione X} with the considerations in \eqref{lin problem plus perturbative term}.\\
At this point we define the following linear operators:
$$ 
\begin{array}{lcr}
\mathcal{U}_0 f (\xi)=\displaystyle\int_{-\infty}^{\xi} e^{- a (\xi-z)}f(z)\d z&,&\mathcal{U}_+ f (\xi)=\displaystyle\int_{-\infty}^{\xi} e^{- b (\xi-z)}f(z)\d z.

\end{array} 
$$

Apply the variation of constants method to \eqref{evolution Y0}--\eqref{evolution Y}, obtaining
$$
\begin{array}{rcl}
Y_0(\xi,\sigma)&=&\mathcal{U}_0 \phi (\xi)-\mathcal{U}_0 \rho_0 (\xi,\sigma),\\
Y_+(\xi,\sigma)&=&\mathcal{U}_+ \phi (\xi)-\mathcal{U}_+ \rho_+ (\xi,\sigma),\\
Y(\xi,K,\sigma)&=&\mathcal{U}_0 \phi (\xi)-\mathcal{U}_0 \rho (\xi,K,\sigma).

\end{array}
$$
Using these equation we can transform \eqref{phi qualcosa} into:
\begin{multline}
\phi(\xi)=m \left[ a  (\mathcal{U}_0 \phi (\xi) - \mathcal{U}_0 \rho_0 (\xi,\sigma))+\rho_0 (\xi,\sigma)\right]+(1-m)\left[ b  (\mathcal{U}_+ \phi (\xi) - \mathcal{U}_+ \rho_+ (\xi,\sigma))+\rho_+ (\xi,\sigma)\right]\\
+m a \int_\mathbb{R} \left(\mathcal{U}_0 \rho_0(\xi,\sigma) -\mathcal{U}_0 \rho (\xi,K,\sigma)\right) Q'(K)\d K+m\int_\mathbb{R} \left(  \rho (\xi,K,\sigma)-\rho_0(\xi,\sigma)\right) Q'(K)\d K.
\end{multline}

Rearranging the terms we obtain the following equation:
\begin{equation}
\mathcal{L}\phi(\xi)=W(\xi),
\label{starting equation}
\end{equation} 
where 
\begin{equation}
\mathcal{L}\phi(\xi)=\phi(\xi)-ma\int_{-\infty}^{\xi} e^{-a(\xi-z)}\phi(z)\d z-(1-m)b\int_{-\infty}^{\xi} e^{-b(\xi-z)}\phi(z)\d z,
\label{definition L}
\end{equation}
and
\begin{multline}\label{W}
W(\xi)=m\left( -a\mathcal{U}_0[\rho_0](\xi,\sigma)+\rho_0(\xi,\sigma)\right)+(1-m)\left(-b\mathcal{U}_+[\rho_+](\xi,\sigma)+\rho_+(\xi,\sigma)\right)+\\
+ma\int_\mathbb{R}\left(\mathcal{U}_0[\rho_0](\xi,\sigma)-\mathcal{U}_0[\rho](\xi,K,\sigma) \right)Q'(K)\d K+m\int_\mathbb{R}\left(\rho(\xi,K,\sigma)-\rho_0(\xi,\sigma) \right)Q'(K)\d K.
\end{multline}

As it is shown in Lemma \ref{invertibilità L} the operator $\mathcal{L}$ defined in \eqref{definition L} can be inverted. Moreover, as proved in Lemma \ref{funzione di green per operatore}, we can explicitly solve the equation
$$
\mathcal{L}G(x)=\delta_a(\xi),
$$
whose solution is given in \eqref{funzione Green}.\\

Thanks to the considerations above, using the Green function $G$ we can write the solutions of \eqref{starting equation} as
\begin{equation}\label{resolutive equation}
\phi(\xi)=\int_{-\infty}^{\xi} G(\eta-\xi)W(\eta) d\eta
\end{equation}
we want to solve this equation using a fixed point argument. To do so recall the space
$$
K_{M,\xi_0}(\delta)=\left\lbrace \phi\in\mathcal{C}^0:|\sigma(\xi)-\sigma_0|=|\phi(\xi)|\leqslant M e^{-(2a+\delta)\xi},M\in\mathbb{R}, \xi\in \left(-\infty,t_0\right]\right\rbrace,
$$
Endowed with the norm
$$
\|\phi\|_{K_{M,\xi_0}(\delta)}=\sup_{\xi\in (-\infty,\xi_0]}|\phi(\xi)|\cdot e^{(2a+\delta)\xi},
$$
we want to apply Banach fixed point theorem for the operator
\begin{equation}\label{operatorT}
T[\phi](\xi)=\int_{-\infty}^{\xi} G(\eta-\xi)W(\eta) \d\eta,
\end{equation}
where the function $W(\xi)=W(\xi,\phi)$, i.e. there is a direct non-linear dependence on the function $\phi$. To do so we proceed as usual in two steps,
\begin{enumerate}
\item Check that $T$ maps $K_{M,T}(\delta)$ to itself for $\xi_0$ sufficiently negative.\label{Banach 1}
\item Check that $T$ is a contraction. \label{Banach 2}
\end{enumerate}
Indeed in order to verify \ref{Banach 1} and \ref{Banach 2} we need some estimate for $|W(\xi)|$, which are given in Proposition \ref{prop estimate W}, i.e. equation \eqref{estimate W} tells us that
\begin{equation*}
|W(\xi)|\leqslant C e^{-2a\xi},
\end{equation*}
for some positive $C=C(a,b,m,\delta)<\infty$ uniformly in $\delta$ and $\xi<\xi_0$ sufficiently negative.\\

This is the first ingredient in order to prove the fixed point for $T$.\\
Since $G\in K_{M,\xi_0}(\delta)$ for each $\delta >0$, then
$$
|T[\phi](\xi)|=\left|\int_{-\infty}^{\xi} G(\eta-\xi)W(\eta)\d\eta\right|\leqslant\\
 CMe^{(2a+\delta)\xi}\int_{-\infty}^{\xi} e^{-(4a+\delta)}\eta \d\eta=-\frac{CM}{4a+\delta}e^{-2a\xi}=-\frac{CMe^{\delta\xi}}{4a+\delta}e^{-(2a+\delta)\xi},
$$
and indeed
$$
\frac{Ce^{\delta\xi}}{4a+\delta}<M,
$$
for $\xi$ sufficiently negative. This proves that the operator $T$ effectively maps $K_{M,\xi_0}(\delta)$ onto itself if $\xi_0$ is sufficiently negative.\\

To prove that $T$ is indeed a contraction on $K_{M,\xi_0}(\delta)$ we will have to repeat some calculations which are made explicit in Proposition \ref{prop deltaW}.\\

Considering the bound given in \eqref{deltaW} we can evaluate
\begin{multline}\label{last estimate}
\left| T[\phi_1](\xi)-T[\phi_2](\xi)\right|=\left|\int_{-\infty}^{\xi} G(\eta-\xi)\left[W(\eta,\phi_1)-W(\eta,\phi_2)\right]\d\eta\right|\\
\leqslant \int_{-\infty}^{\xi} |G(\eta-\xi)||\left[W(\eta,\phi_1)-W(\eta,\phi_2)\right]|\d\eta\lesssim \left\|\phi_1-\phi_2\right\|_{K_{M,\xi_0}(\delta)}\int_{-\infty}^{\xi} |G(\eta-\xi)|e^{-(2a+\delta)\eta}\d\eta\\
\overset{G\in K_{M,\xi_0}(\delta)}{\leqslant} \left\|\phi_1-\phi_2\right\|_{K_{M,\xi_0}(\delta)} \int_{-\infty}^{\xi} e^{-(4a+2\delta)\eta}d\eta\lesssim \left\|\phi_1-\phi_2\right\|_{K_{M,\xi_0}(\delta)}\cdot e^{-2(2a+\delta)\xi}.
\end{multline}
Multipling both sides of \eqref{last estimate} for $e^{(2a+\delta)\xi}$, and taking the sup for $\xi\leqslant\xi_0$ we obtain
$$
\left\| T[\phi_1](\xi)-T[\phi_2](\xi)\right\|_{K_{M,\xi_0}(\delta)}\lesssim e^{-(2a+\delta)\xi_0}\left\|\phi_1-\phi_2\right\|_{K_{M,\xi_0}(\delta)},
$$
which guarantees that $T$ is a contraction concluding the proof of the theorem.
 \end{proof}

\section{Global well posedness.}\label{globalExistence}
From now on the re-scaled time $\xi$ is going to be called $t$.\\
 The following lemma is the starting point of our analysis.
\begin{lemma} \label{sigma as characteristic}
Consider $\sigma$ as defined in \eqref{definizione sigma}, than, $\sigma$ is completely determined by the evolution of the characteristics $X_\pm$.
\end{lemma}
\begin{proof}
Apply integration by parts obtaining 
\begin{equation} \label{sigma 1}
\sigma(t)=H'(X_+(t))-\int_{X_-(t)}^{X_+(t)}H''(x)R(x,t)\d x.
\end{equation}
In the same way we can obtain the following equation
\begin{equation}
\label{sigma 2}
\sigma(t)=H'(X_-(t))+\int_{X_-(t)}^{X_+(t)}H''(x)\left(1-R(x,t)\right)\d x.
\end{equation}
\end{proof}
In view of the simple computations in Lemma \ref{sigma as characteristic} it is equivalent to prove global existence for $ \sigma $ or for the characteristics $X_\pm$. Recall that
$$
\left\lbrace
\begin{array}{l}
X_+'(t)=-H'(X_+(t))+\sigma(t),\\
X_-'(t)=-H'(X_-(t))+\sigma(t),
\end{array}
\right.
$$
considering \eqref{sigma 1} and \eqref{sigma 2} we obtain the following new differential equations for the characteristics
\begin{equation}\label{system for char}
\left\lbrace
\begin{array}{l}
X_+'(t)=-\displaystyle\int_{X_-(t)}^{X_+(t)}H''(x)R(x,t)\d x,\\[3mm]
X_-'(t)=\displaystyle\int_{X_-(t)}^{X_+(t)}H''(x)\left(1-R(x,t)\right)\d x.
\end{array}
\right.
\end{equation}
Which is very interesting since we see explicitly in \eqref{system for char} that the evolution of $X_\pm$ is not influenced by $\sigma$, hence we can write \eqref{system for char} as $X'=f(X,t)$ with $f:\mathbb{R}^2\times\mathbb{R}\rightarrow\mathbb{R}^2$ and has the following explicit formulation
\begin{equation}
\label{f sistema}
f(x,t)=\left(\begin{array}{lcr}
-\displaystyle\int_{x_2}^{x_1}H''(x)R(x,t)\d x&,&
\displaystyle\int_{x_2}^{x_1}H''(x)\left(1-R(x,t)\right)\d x
\end{array}\right).
\end{equation}

At this point we can start to study the system \eqref{system for char}.
\begin{prop} \label{global existence for char}
Under the assumptions \eqref{simple hypotesis on nonlinearity} the system \eqref{system for char} admits unique solution globally in $\mathbb{R}$.
\end{prop}
\begin{proof}
The function $f$ is Lipschitz continuous and sub-linear in the variable $x$, hence apply \cite[Theorem 2.17]{teschl} obtaining our global result.
\end{proof}

\

The global existence statement proved in Proposition \ref{global existence for char} gives, thanks to the computations performed in Lemma \ref{sigma as characteristic} a global existence result also for the function $\sigma$, which we showed that the evolution of $\sigma$ can be completely described in terms of the evolution of the characteristics $X_\pm$.\\
We would like to refine even further this result, namely we would like to be able to give some $L^\infty$ bound for the characteristics $X_\pm$ and, consequently, for the function $\sigma$ that we might need in the future.\\

\begin{lemma}
Define $\Delta X(t)=X_+(t)-X_-(t)$, and consider a potential $H$ which satisfies \eqref{simple hypotesis on nonlinearity}. In this case $\Delta X\in L^\infty (\mathbb{R})$.
\end{lemma}

Now we can obtain the $L^\infty$ estimates for $\sigma,X_\pm$.
\begin{prop}\label{L inf bounds for char}
Consider a potential $H$ satisfying \eqref{simple hypotesis on nonlinearity}, than we have that the function $\sigma$ defined in \eqref{sigma 1} or equivalently \eqref{sigma 2} and the functions $X_\pm$ described by the system \eqref{system for char} are not only defined in all $\mathbb{R}$ but they also belong to $L^\infty(\mathbb{R})$.
\end{prop}
\begin{proof}
The proof is, at this point, very short. Consider $\sigma$ as in \eqref{sigma 1}, namely 
$$
\sigma(t)=H'(X_+(t))-\int_{X_-(t)}^{X_+(t)}H''(x)R(x,t)\d x.
$$
We know thanks to the previous lemma that $|\Delta X|<L$ hence $\left|\int_{X_-(t)}^{X_+(t)}H''(x)R(x,t)\d x\right|\leqslant\|H\|_{L^\infty}\cdot L< \infty$. Which means that $\int_{X_-(\bullet)}^{X_+(\bullet)}H''(x)R(x,\bullet)\d x\in L^\infty(\mathbb{R})$. Considering \eqref{simple hypotesis on nonlinearity}, $H'(x)=\alpha x+g(x)$ where $g$ is a $L^\infty$ function. \\
With these considerations we obtain that
$$
\sigma(t)=\int H'(x) \rho dx =\alpha \int x \rho \d x+\int g (x) \rho \d x\leqslant \alpha \ell^\star +\|g\|_{L^\infty(\mathbb{R})}<\infty
$$

We've obtained hence that $\sigma,\int_{X_-(\bullet)}^{X_+(\bullet)}H''(x)R(x,\bullet)\d x\in L^\infty$, which means, considering \eqref{sigma 1}, that $H'(X_+)\in L^\infty$ which implies that $X_+\in L^\infty$.\\
To prove that $X_-\in L^\infty$ the reasoning is the same considering $\sigma$ as in the equation \eqref{sigma 2}.
\end{proof}
\section{Stability for $t\to +\infty$.}\label{stability}
Proposition \ref{well posedness at -infty} assures us that as long as $t\to -\infty$ the solution to \eqref{mass splitting equation} stabilize to a convex combination of Dirac-$\delta$ measures localized in two points. We expect to recover, after the mass splitting process, a similar configuration.\\
To do so we are going to show that the system effectively converges to such a form via a standard stability argument  in a neighbourhood of $t=+\infty$.
\begin{lemma}
Define the following function
\begin{equation}
\label{energy}
E_t=\int H(x)\rho \d x<\infty,
\end{equation}
where $\rho=\rho (t,x)$ solution of \eqref{mass splitting equation}. Then we have that $E_t $ is decreasing in time.
\end{lemma}
\begin{proof}
Derive equation \eqref{energy}, integrate the obtained equation by parts and apply Jensen inequality obtaining that
$$
\frac{d}{dt}\left(\int H(x)\rho \d x\right)\leqslant 0
$$

Which is exactly what we wanted to prove.
\end{proof}

At this point it is clear that as $t\to +\infty$ we have that the system stabilize to some value $E_{+\infty}\leqslant E_t$ for each $t\in \mathbb{R}\cup\{-\infty\}$, and that \eqref{energy} holds.\\

\begin{theorem}\label{limit rho}
There exists some sequence sequence $\left( t_m\right)_{m\in\mathbb{N}}$ such that $t_m\xrightarrow[]{m\to\infty}+\infty$. Set $\rho_m(x)=\rho(t_m,x)$, than there exist some $x_-,x_0,x_+$, such that $x_\pm\in\left\lbrace H''\geqslant 0\right\rbrace$ and $x_0\in\left\lbrace H''< 0\right\rbrace$ and
\begin{equation}
\label{conv via sottosucc}
\rho_m(x)\xrightharpoonup[]{\star}\sum_{i\in\left\lbrace-,0,+\right\rbrace}m_i\delta_{x_i}(x)
\end{equation}
with
$
m_-+m_0+m_+=1
$
and
\begin{equation}\label{samepoints}
H'(x_-)=H'(x_+)=H'(x_0)
\end{equation}
\end{theorem}

\begin{proof}
To prove the theorem  proceed as follows. Define the \textbf{dissipation} of the system as
\begin{equation}
\label{dissipation}
-D(t)=\frac{d}{dt}\left(\int H(x)\rho dx\right)=-\int (H'(x))^2\rho dx + \left(\int H'(x)\rho dx\right)^2\leqslant 0
\end{equation}
From this definition we obtain that
$
\int_{-\infty}^TD(t)dt=-\int_{-\infty}^T\frac{\d}{\d t}E_t\d t=E_{-\infty}-E_T<E_{-\infty}.
$

At this point, hence, we have obtained that
\begin{equation}
\label{integrabilità D}
\int_\mathbb{R}D(t)dt<\infty,
\end{equation}
which implies that there exists a sequence $\left(t_m\right)_m$ such that $D(t_m)\to 0$ as $m\to +\infty$. In particular let us select some sequence $t_m\xrightarrow{m\to\infty}+\infty$, we have that $D(t_m)\to 0$, i.e.
\begin{equation}\label{limit measures}
\int \left( H'(x)\right)^2\rho_m dx-\left(\int H'(x)\rho_m dx\right)^2\xrightarrow{m\to\infty}0.
\end{equation}

Where $\rho_m(x)=\rho(t_m,x)$.\\
Moreover the set $\left(\rho(t,\cdot)\right)_t$ is bounded in $rca(\mathbb{R})$, the set of regular measures, with finite variation. By Banach-Alaoglu theorem is weak-$\star$ compact, hence there exists a (not relabelled) sequence of diverging times $\left( t_m\right)_m$ such that $\rho_m \overset{\star}{\rightharpoonup} \rho$.\\
\eqref{limit measures} verifies if and only if $H'=k$ for some $k\in\mathbb{R}$ $\rho$--almost everywhere, for $\rho=\star-\lim_m \rho_m$. This consideration, together with \eqref{simple hypotesis on nonlinearity}, and the fact that $H'$ is strictly monotone in its invertible branches, implies that
$$
\rho_m\overset{\star}{\rightharpoonup}\sum_{i\in\left\lbrace-,0,+\right\rbrace}m_i\delta_{x_i},
$$
where the points $x^+_i,i\in\left\lbrace -,0,+\right\rbrace$ have to satisfy the condition \eqref{samepoints} since $H'=k$ $\rho$-a.s. or, else
$
\rho_m \overset{\star}{\rightharpoonup} \delta_{\bar{x}}
$
for some $\bar{x}\in\left[x_\star,x^\star\right]^c$.\\
A priori could as well happen that $\rho_m \overset{\star}{\rightharpoonup} \delta_\infty$, we want to exclude this eventuality.\\
To do so  consider 
$
\sigma (t_m)=\int H'(x)\rho_m(x)\d x\to\infty
$
but $\sigma\in L^\infty$ as proved in Proposition \ref{L inf bounds for char}. hence this is absurd.
\end{proof}

\begin{prop}\label{convergenza debole migliorata}
Set a sequence $\left(t_m\right)_m$ such that $\rho_m\overset{\star}{\rightharpoonup}\sum_{i\in\left\lbrace-,0,+\right\rbrace}m_i\delta_{x_i}$. Than for any $f\in\mathcal{C}(D)$ where $D$ is the following compact set of $\mathbb{R}$: $D=\left[D_-,D_+\right]$ with
\begin{align*}
D_-&=\inf_{t\in\mathbb{R}} X_-(t),\\
D_+&=\sup_{t\in\mathbb{R}}X_+(t),
\end{align*}
where $X_\pm$ are solution of \eqref{system for char}, the following equality holds
$$
\lim_{m\to\infty}\int f(x)\rho(t_m,x)\d x =\sum_{i\in\{-,0,+\}}m_if(x_i).
$$
\end{prop}
\begin{proof}
We point out at first that $|D_\pm|<\infty$ thanks the fact that $X_\pm\in L^\infty$ as shown in Proposition \ref{L inf bounds for char}.\\
With these considerations, for every time $t$ the mass of the entire system is concentrated in the set $\left[X_-(t),X_+(t)\right]$, this has been described in detail in Lemma \ref{dimostrazione condizione H}, considering also the global result in Proposition\ref{global existence for char}. This implies that the support of the probability density $\rho$ is compact, and independent from the time $t$, namely the set $D$. The claim hence follows since, in this setting,
$
\left(\mathcal{C}^0\left(D\right)\right)^\star=rca(D)
$
and $\rho(t,\cdot)\in  rca(D)$ for all $t\in\mathbb{R}$.
\end{proof}
\begin{cor}\label{convergenza debole per classi di funzioni}
Set a sequence $\left(t_m\right)_m$ such that $\rho_m\overset{\star}{\rightharpoonup}\sum_{i\in\left\lbrace-,0,+\right\rbrace}m_i\delta_{x_i}$.
Then
\begin{align*}
\lim_{m\to\infty}\sigma\left(t_m\right)=H'\left(x_i\right)& & i\in\left\lbrace-,0,+\right\rbrace\\
\ell=\sum_{i\in\left\lbrace-,0,+\right\rbrace}m_i x_i 
\end{align*}
\end{cor}
\begin{proof}
Apply Proposition \ref{convergenza debole migliorata} to the functions $\sigma\left(t\right)=\int H'(x)\rho(t,x)\d x$ and $\ell^\star=\int x \rho\left(t,x\right)\d x$ obtaining
\begin{equation}
\label{approx sigma sottosucc}
\sigma(t_m)=\sum_{i\in\left\lbrace-,0,+\right\rbrace}m_i H'(x_i) +o(1)=H'(x_i)+o(1).
\end{equation}
and
\begin{equation}\label{appelle}
\ell=\int x\rho_m(x)\d x=\sum_{i\in\left\lbrace-,0,+\right\rbrace}m_i x_i ,
\end{equation}
\end{proof}
\begin{cor}\label{L1 migliorata}
For any $t\in\mathbb{R}$ and any $f\in\mathcal{C}(D)$, it is true that $f\in L^1\left(\mathbb{R},\rho(t,\cdot)\right)$.
\end{cor}
\begin{proof}
$
\int f(x)\rho(x,t)\d x\leqslant \max_{x\in D}|f(x)|<\infty.
$
\end{proof}
\begin{lemma}\label{convergenza masse}
Select a sequence $\left(t_m\right)_m$ such that $\rho_m\overset{\star}{\rightharpoonup}\sum_{i\in\left\lbrace-,0,+\right\rbrace}m_i\delta_{x_i}$. Suppose $x_i\neq x_\star,x^\star$. Define the functions 
\begin{align*}
m_0(t)&=\int_{x_\star}^{x^\star}\rho(t,x)\d x,\\
m_-(t)&=\int^{x^\star}_{-\infty}\rho(t,x)\d x,\\
m_+(t)&=\int_{x^\star}^{\infty}\rho(t,x)\d x,
\end{align*}
then $m_i(t_m)\to m_i$.
\end{lemma}

\begin{rem}
Suppose that $x_-=x_0=x_\star$, hence
$
\rho_m\overset{\star}{\rightharpoonup} (m_0+m_-)\delta_{x_\star}+m_+\delta_{x^{\star\star}}.
$
The lemma above proves that
$
m_+=\lim_{m}m_+(t_m),
$
but does not gives any information about the masses $m_{-/0}$.\\
Anyway, since $m_-+m_0+m_+=1$ and for every $t_m$ the relation $m_-(t_m)+m_0(t_m)+m_+(t_m)=1$ still holds, hence we can assert that
$$
m_-+m_0=\lim_m \left[m_-(t_m)+m_0(t_m)\right],
$$
which is going to suffice for the purposes of our analysis.
\hspace*{\fill}$\blacklozenge$\medskip
\end{rem}
After all these considerations we have in particular obtained that
$$
\sigma(t_m)\to H'(x_i),
$$
which allows us to choose a triple $\left( X_-(t_m),X_0(t_m),X_+(t_m)\right)_m$ such that
\begin{equation}
\label{nuovo sigma}
\sigma(t_m)=H'(X_i(t_m)),
\end{equation}

and $X_i(t_m)\to x_i$ by construction. Considering moreover equation \eqref{appelle} and Lemma \ref{convergenza masse} we can hence say that
$$
\ell^\star=\sum_{i\in\{-,0,+\}}m_i(t_m)X_i(t_m) +o(1),
$$
with $o(1)\to 0$ as $t_m\to\infty$.\\
From now on we are going to refer as $X_i$ to functions that satisfy \eqref{nuovo sigma}, and not any more for solutions of \eqref{equation X} with $K=0,\pm\infty$.

\begin{rem}
a priori we could have different weak-$\star$ limit equilibria for $\rho$ depending  on the diverging sequence of times.\\
In other words, set $(t_m),(t'_m)$ different sequences st $t_m,t'_m\to +\infty$, set $\rho_m,\rho'_m$ the distributions, the above theorem explains that $\rho_m\overset{\star}{\rightharpoonup}\Sigma,\rho'_m\overset{\star}{\rightharpoonup}\Sigma'$, but might happen, a priori, that $\Sigma\neq\Sigma'$. The next question is: is there uniqueness at the limit? If yes under which conditions?\\

\hspace*{\fill}$\blacklozenge$\medskip
\end{rem}
\subsection{The uniqueness problem.}\label{uniqueness}
In this section we want to identify some hypothesis under which the problem treated all along this paper has a unique weak-$\star$ limit as $t\to\infty$.\\

The first step in our analysis is going to be the following lemma
\begin{lemma}
Consider $D$ defined as in \eqref{dissipation}. Than 
$\lim_{t\to +\infty}D(t)=0.
$
\end{lemma}
\begin{proof}
Thanks to Corollary \ref{L1 migliorata} we can argue that
\begin{equation}
\label{qualcosa che uso}
\int\left( H'(x)\right)^2 \rho(x,t)\d x\leqslant\max_D| H'|^2
<\infty.
\end{equation}

With this consideration in hand, considering that $D$ is defined as in \eqref{dissipation}, it is easy to check that $D,D'\in L^\infty(\mathbb{R})$.  Accordingly to \eqref{dissipation}
$$
D(t)=-\underbrace{\left(\int H'(x)\rho(x,t)\d x\right)^2}_{=\sigma^2(t)\in L^\infty}+\underbrace{\int\left(H'(x)\right)^2\rho(x,t)\d x}_{\overset{\eqref{qualcosa che uso}}{<}\infty}\in L^\infty.
$$

And 
\begin{multline*}
D'(t)=-2\sigma(t)\cdot \int H'(x)\rho_t(x,t)\d x+\int\left(H'(x)\right)^2\rho(x,t)\d x\overset{\text{IbP}+\eqref{mass splitting equation}}{=}\\
2\sigma(t)\cdot\int H''(x)\left(H'(x)-\sigma(t)\right)\rho(x,t)\d x+\int\left(H'(x)\right)^2\rho(x,t)\d x\leqslant\\
2\|\sigma\|_{L^\infty}\left(\max_D |H'\cdot H''|+\|\sigma\|_{L^\infty}\right)+\max_D |H'|^2<\infty.
\end{multline*}

Where in the last inequality we applied Corollary \ref{L1 migliorata}.\\
At this point, suppose it is not true that $\lim_{t\to\infty}D(t)=0$. Hence there exists a sequence such that $D(t_n)>\varepsilon$ for some $\varepsilon >0$. Define the set
$
A_\varepsilon=\left\lbrace t:D(t)\geqslant\varepsilon\right\rbrace.
$, 
and the set $A_{\varepsilon,n}=A_{\varepsilon}\cap[n,\infty)$. Moreover
$
A_{\varepsilon,n}=\bigcup_kA_{\varepsilon,n}^k,
$
where $A_{\varepsilon,n}^k$ are the connected components of $A_{\varepsilon,n}$, indexed by $k$. It is true that $\mathcal{L}\left(A_{\eta,n}\right)\to 0$ as $n\to\infty$, otherwise $D$ wouldn't be $L^1\left(\mathbb{R},\d t\right)$ contradicting \eqref{integrabilità D}, this implies that also $\mathcal{L}\left(A_{\eta,n}^k\right)\to 0$, for every index $k$. This consideration is valid for every $\eta >0$, in particular for $\varepsilon/2$.\\
Select a $t_n\in A_{\varepsilon,n}^k,s_n\notin A_{\varepsilon/2,n}$. These two sequences can be selected in such a way that 
$
|t_n-s_n|\xrightarrow{n\to\infty}0.
$ 
 This is indeed true since $\left(s_n\right)_n$ can be chosen such that 
$
|t_n-s_n|\leqslant 2\cdot \text{diam}\left(A_{\varepsilon/2,n}^k\right)\xrightarrow{n\to\infty}0,
$
infering via Lagrange theorem to state that there exist a sequence $\left(\tau_n\right)_n$ such that $|D(t_n)-D(s_n)|=|D'(\tau_n)||t_n-s_n|$, but since $t_n\in A_{\varepsilon,n}^k,s_n\notin A_{\varepsilon/2,n}$ is easily obtained that 
$
|D(t_n)-D(s_n)|\geqslant\frac{\varepsilon}{2},
$
 but, considering that $|t_n-s_n|\xrightarrow{n\to\infty}0$ this would imply that $|D'(\tau_n)|\to\infty$, contradicting  $D'\in L^\infty$ and concluding the proof.
\end{proof}

This lemma states a very important property, which is that, for every diverging sequence of times we can extract a subsequence such that $\rho_m\overset{\star}{\rightharpoonup}\sum_{i\in\left\lbrace-,0,+\right\rbrace}m_i\delta_{x_i}$.\\

We would like to understand better the structure of the invertible branches$X_i$ of $H$.
\begin{lemma}\label{asymptotic expansion of the invertible branches}
Assume the potential $H$ is $ \mathcal{C}^4_{loc}\left(\mathbb{R}\right) $, and consider a neighborhood $\left(\sigma^\star-\delta,\sigma^\star\right]$, where $\delta$ is considered to be small. Recall that, accordingly to \eqref{nuovo sigma} the functions $X_i,i=-,0,+$ represent the invertible branches of the potential $H$. Then the functions $A_{-/0}\left(\sigma\right)=X_{-/0}\left(\sigma\right)-x_\star$ are uniquely determined,  $A_{-/0}\left(\sigma\right)=\mathcal{O}\left(\sqrt{\sigma^\star-\sigma}\right)$ in a vicinity if $\sigma^\star$, and, moreover, setting $\sigma^\star-\sigma=\Delta\sigma$ under the regularity assumption made on $H$, and supposing that $H'''\left(x_\star\right)\neq 0$ the following expansion holds
$$
 A_{-/0}(\sigma)=\sum_{n= 1}^3a_n^{-/0}\left(\Delta\sigma\right)^{n/2}+\mathcal{O}
 \left(\left(\Delta\sigma\right)^2\right).
$$

On the other hand, always in the same vicinity of $\sigma^\star$, denoting accordingly to the notation introduced, $X_+\left(\sigma^\star\right)=\sigma^{\star\star}$ we have  $X_+\left(\sigma^\star\right)-x^{\star\star}=\mathcal{O}\left(\Delta\sigma\right)$, and 
$$
X_+\left(\sigma^\star\right)-x^{\star\star}=\sum_{n= 1}^3a^+_n\left(\Delta\sigma\right)^n
+\mathcal{O}
 \left(\left(\Delta\sigma\right)^2\right).
$$

Moreover $\left| A_{-/0}\left(\sigma\right)\right|,\left|X_+\left(\sigma^\star\right)-x^{\star\star}\right| >0$  if $\sigma\neq \sigma^\star$.
\end{lemma}

\begin{proof}
The uniqueness is clear and comes from the definition in \eqref{nuovo sigma} considering that, by hypothesis, we have been considering potentials $H$ with strictly monotone invertible branches.\\
By definition of the function $X_0$ (see \eqref{nuovo sigma}) we have that $\sigma=H'\left(X_0\left(\sigma\right)\right)$, performing a Taylor expansion of the right hand side of this equation in terms of the perturbation $X_0\left(\sigma\right)-x_\star$ and considering the fact that $H''(x_\star)=0$ we obtain that
 \begin{equation}\label{eq0}
 \sigma-\sigma^\star=\frac{H'''(x_\star)}{2}\left(X_0(\sigma)-x_\star\right)^2+\mathcal{O}\left(\left(X_0(\sigma)-x_\star\right)^3\right).
 \end{equation}

Now, from equation \eqref{eq0} we can assert that $H'''(x_\star)\leqslant 0$ comparing the signs of the left hand side with the right hand side, moreover, considering that by hypothesis  $H'''\left(x_\star\right)\neq 0$ we obtain that $H'''\left(x_\star\right)< 0$.\\
We need a detailed analysis of the factors $A_{-/0}$, where  $A_{-/0}(\sigma)=X_{-/0}(\sigma)-x_\star$. Let us make the following ansatz
 \begin{equation}
 \label{espansione Azero/meno}
 A_{-/0}(\sigma)=\sum_{n= 1}^3a_n^{-/0}(\Delta\sigma)^{n/2}+\mathcal{O}
 \left(\left(\Delta\sigma\right)^2\right),
  \end{equation}

which we will justify at the end of this proof, see Remark \ref{fz implicita}, $\Delta\sigma=\sigma^\star-\sigma$. \\

 From \eqref{eq0}, that
 $$
 A_0(\sigma)= \sqrt{\frac{2}{|H'''(x_\star)|}}(\sigma^\star-\sigma)^{1/2}+o\left((\sigma^\star-\sigma)^{1/2}\right),
 $$
 The same procedure gives us that
  $$
 A_-(\sigma)= -\sqrt{\frac{2}{|H'''(x_\star)|}}(\sigma^\star-\sigma)^{1/2}+o\left((\sigma^\star-\sigma)^{1/2}\right),
 $$
 where indeed we have that $ A_-(\sigma)=X_-(\sigma)-x_\star$. What is left is to understand the asymptotic behavior of the linear term $X_+\left(\sigma^\star\right)-x^{\star\star}$, but this is easily obtained performing the same procedure above, in particular we obtain
  \begin{equation}
 \label{correzione 3}
 X_+(\sigma)= X_+(\sigma^\star)+\frac{\sigma-\sigma^\star}{H''(x^{\star\star})}+o\left(\sigma-\sigma^\star\right),
 \end{equation}
 hence there is a linear dependence from the parameter $\sigma$, since we do know that $H''(X_+(\sigma^\star))>0$. Equation \eqref{correzione 3} can be justified formally via an argument similar to the one performed in Remark \ref{fz implicita}.\\
 Putting together the results obtained we get
 \begin{equation}
  \label{stimesigma}
 \begin{array}{l}
 X_+(\sigma)= X_+(\sigma^\star)-b(\sigma-\sigma^\star)+o\left(\sigma-\sigma^\star\right),\\
 X_0(\sigma)= x_\star +c (\sigma^\star-\sigma)^{1/2}+o\left((\sigma^\star-\sigma)^{1/2}\right),\\
 X_-(\sigma)= x_\star -c (\sigma^\star-\sigma)^{1/2}+o\left((\sigma^\star-\sigma)^{1/2}\right).
 \end{array}
 \end{equation}
 Where $c=\sqrt{\frac{2}{|H'''(x_\star)|}},b=\frac{1}{H''(x^{\star\star})}$.\\
 
 At this point the first order expansion is clear. We will need though in the following the expansion of $A_{-/0}$ in \eqref{espansione Azero/meno} up to the linear term, i.e. the second order.\\
 To do so we evaluate the next Taylor element in \eqref{eq0} we obtain that
 \begin{equation}
 \label{approxdeltasigmazero}
\Delta \sigma = \frac{1}{2}H'''(x_\star)A_0(\sigma)^2+\frac{1}{6}H^{(4)}(x_\star)A_0(\sigma)^3+ \mathcal{O}\left(A_0(\sigma)^4\right) \\=c_1A_0(\sigma)^2+c_2A_0(\sigma)^3+ \mathcal{O}\left(A_0(\sigma)^4\right),
\end{equation}
indeed moreove we can express $A_0(\sigma)$ as
 \begin{equation}
 \label{approssimazioneAzero}
 A_0(\sigma)= a_1^0(\Delta\sigma)^{1/2}+a_2^0\Delta\sigma +\mathcal{O}\left((\Delta\sigma)^{3/2}\right),
 \end{equation}
 plugging \eqref{approssimazioneAzero} into \eqref{approxdeltasigmazero} and after some algebraic manipulation  we obtain
 $$
 \Delta\sigma = \left( a_1^0 \right)^2c_1\Delta\sigma+\left( \left(a_1^0\right)^3c_2+2a_1^0a_2^0c_1\right)\left(\Delta\sigma\right)^{3/2}+o\left((\Delta\sigma)^{3/2}\right),
 $$
 equating the coefficients of $\Delta\sigma$ and $\left(\Delta\sigma\right)^{3/2}$ to zero and solving the non-linear system in the unknown $a_i^0$ we obtain two solutions, namely the two couples $\left(a_1^0,a_2^0\right)$ and $\left(a_1^-,a_2^-\right)$
 
 \begin{align} \label{coeff1}
 a_1^0&=\sqrt{\frac{1}{c_1}},&a_1^-&=-\sqrt{\frac{1}{c_1}},\\
 a_2^0&=-\frac{c_2}{2c_1^2},& a_2^-&=-\frac{c_2}{2c_1^2},\label{coeff 2}
  \end{align}
  We recall that $c_1,c_2$ are defined as in \eqref{approxdeltasigmazero}.\\
\end{proof}

\begin{rem}\label{fz implicita}
We want to justify equation \eqref{espansione Azero/meno}. Taylor expantion yelds 
$
\Delta\sigma=cA_{-/0}(\sigma)^2+\mathcal{O}\left(A_{-/0}(\sigma)^3\right),
$
with $c>0$. For $\sigma$ sufficiently close to $\sigma^\star$ both right and left hand side of the equation above are positive, hence it makes sense to take the square root on both sides obtaining the following two equations
\begin{align*}
\left(\Delta\sigma\right)^{1/2}&=A_0\sqrt{c+\mathcal{O}\left(A_0\right)},\\
\left(\Delta\sigma\right)^{1/2}&=-A_-\sqrt{c+\mathcal{O}\left(A_-\right)}.
\end{align*}
Define 
$
F^0\left(\left(\Delta\sigma\right)^{1/2},A_0\right)=A_0\sqrt{c+\mathcal{O}\left(A_0\right)}-\left(\Delta\sigma\right)^{1/2},
$
a straightforward computation shows that
$
\frac{\partial F^0}{\partial A_0}\left(\left(\Delta\sigma\right)^{1/2},0\right)=\sqrt{c}\neq 0.
$
By the Implicit function theorem we argue that there exists a function $A_0=A_0\left(\left(\Delta\sigma\right)^{1/2}\right)$ such that,
$
\Delta\sigma=cA_0\left(\left(\Delta\sigma\right)^{1/2}\right)^2+
\mathcal{O}\left(A_0\left(\left(\Delta\sigma\right)^{1/2}\right)^3\right),
$
moreover $A_0$ was a $ \mathcal{C}^4_\text{loc}\left(\mathbb{R}\right) $, hence we can express it as
$
A_0\left(\left(\Delta\sigma\right)^{1/2}\right)=\displaystyle\sum_{n=
1}^3a_n^0 \left(\Delta\sigma\right)^{n/2}+\mathcal{O}
 \left(\left(\Delta\sigma\right)^2\right),
$
which is exactly \eqref{espansione Azero/meno}. A similar approach is valid also for $A_-$.\hspace*{\fill}$\blacklozenge$\medskip
\end{rem}

At this point we can study the problem of the uniqueness as $t\to\infty$ for potentials $H$ which are $\mathcal{C}^4_\text{loc}\left(\mathbb{R}\right)$.

\begin{prop}\label{effetto fisa}
Suppose that the potential $H$ satisfies the hypotesis stated at the beginning of this paper, in Subsection \ref{assumptions on the potential}, and moreover  suppose that $\sigma(t)\in\left(\sigma_\star+\eta,\sigma^\star-\eta\right)$ for some $\eta >0$ and for all times $t>t_0$. Then the limit is unique.
\end{prop}
\begin{proof} 
First of all we claim that, there exist a local maximum $M_0(t)$ of $\rho(x,t)$, around which all the mass of the unstable region concentrates such that, 
$
\lim_{t\to\infty}\bigl|M_0(t)-X_0(t)\bigr|=0,
$
where $X_0(t)$ is the only spinodal state such that $\sigma(t)=H'\left(X_0(t)\right)$. $M_0$ exists thanks to Theorem \ref{limit rho}.\\

Moreover we claim that $M_0\in\mathcal{C}^1(\mathbb{R})$. This in indeed true since we are considering the mass transported along characteristics, and by Proposition \ref{global existence for char} these are defined and $\mathcal{C}^1$ globally in $\mathbb{R}$.\\
Since $M_0$ is a local maximum for $\rho$ satisfies
$
\rho_x\left(M_0(t),t\right)=0,
$
and hence, expanding the equation \eqref{mass splitting equation} into $\rho_t=H''(x)\rho+\left(H'(x)-\sigma(t)\right)\rho_x$ we obtain that
\begin{equation}\label{ennesima roba inutile}
\rho_t\left(M_0(t),t\right)=H''(M_0(t))\rho\left(M_0(t),t\right).
\end{equation}
Set $\rho\left(M_0(t),t\right)=r(t)$. Thanks to the hypothesis $\sigma(t)\in\left(\sigma_\star+\eta,\sigma^\star-\eta\right)$, and since $M_0=X_0+o(1)$ we can state that $M_0(t)\in\left(x_\star+\eta',x^\star-\eta'\right)$ for some $\eta'$ small and for $t$ sufficiently large. By the structure of the potential, hence
\begin{equation}\label{roba con epsilon}
H''(M_0(t))\leqslant -\epsilon,\epsilon>0,
\end{equation}
hence, considering \eqref{roba con epsilon} we obtain
$
r'(t)\leqslant-\epsilon r(t).
$
We can, at this point, apply Gronwall inequality to the previous inequality, obtaining 
$
r(t)\leqslant r(t_0)e^{-\epsilon(t-t_0)}\xrightarrow{t\to\infty}0,
$
Proving that $m_0(t)\to 0$ as $t\to\infty$.\\
Moreover, under these assumptions, $\dot{m}_\pm >0$, hence
\begin{equation}\label{appellemzerovanishes}
\ell^\star=\left(m_-^\infty +o(1)\right)X_-(\sigma)+\left(m_+^\infty +o(1)\right)X_+(\sigma)+o(1)X_0(\sigma),
\end{equation}
where
$
m_\pm^\infty=\lim_{t\to\infty}m_\pm (t),
$
but, by construction $X_i\in L^\infty$, hence equation \eqref{appellemzerovanishes} can be restated as
\begin{equation}
\label{appellemzerovanishes2}
\ell^\star=m_-^\infty X_-(\sigma)+m_+^\infty X_+(\sigma)+o(1).
\end{equation}

We can, at this point, to prove the uniqueness. Suppose we do not have uniqueness for $\rho$, than there's no uniqueness for $\sigma$ either, which means that $\sigma$ oscillates between two values $\sigma_1<\sigma_2$, and hence at these two values the convex combination for $\ell^\star$ in \eqref{appellemzerovanishes2} would give two different results, in particular 
$$\ell^\star+o(1)=m_-^\infty X_-(\sigma_1)+m_+^\infty X_+(\sigma_1)<m_-^\infty X_-(\sigma_2)+m_+^\infty X_+(\sigma_2)=\ell^\star+o(1),$$
 which is indeed absurd being $\ell^\star$ a conserved quantity.
\end{proof}
\begin{rem}
It might be interesting to notice that Proposition \ref{effetto fisa} underlines clearly the equivalence between the vanishing of the spinodal mass and the uniqueness of the weak-$\star$ limit. \hspace*{\fill}$\blacklozenge$\medskip
\end{rem}
\begin{theorem}\label{unicità 2}
Under assumptions of Subsection \ref{assumptions on the potential} on the potential $H$, if $\ell^\star\in\left[x_\star,x^\star\right]^c$, if $H$ is $\mathcal{C}^4$ around $x_\star$ and $x^\star$, and $H^{(4)}(x_\star),H^{(4)}(x^\star)>0$, the limit is unique.
\end{theorem}
\begin{rem}
In the following proof $o(1)$ is going to be a general perturbation depending only on $t$ such that $o(1)\xrightarrow{t\to\infty}0$.
\end{rem}
\begin{proof}Define
$\overline{\sigma}=\limsup_{t\to\infty}\sigma(t),
\underline{\sigma}=\liminf_{t\to\infty}\sigma(t),
$and consider the case in which $\ell^\star >x^\star$. The other case is simply symmetric. Note that this in particular implies that $\underline{\sigma}>\sigma_\star$.\\

We are going to divide the problem in several simple sub-cases.
\begin{enumerate}
\item suppose $\underline{\sigma}\geqslant\sigma^\star$. Than $\ell^\star=X_+(\sigma)+o(1)$, and hence, since $\ell^\star$ has to be constant $\sigma$ has to converge to a unique limit.
\item the case in which $\overline{\sigma}<\sigma^\star$ has already been discussed in detail in Proposition \ref{effetto fisa}.
\item suppose $\underline{\sigma}<\sigma^\star$ and $\overline{\sigma}>\sigma^\star$. This implies that we can choose an arbitrary large time $t_\star$ such that $\sigma(t_\star)=\sigma^\star$ and such that there exist $t_\star<t_1<t_2$ such that $\sigma^\star<\sigma_1=\sigma(t_1)<\sigma_2=\sigma(t_2)<\overline{\sigma}$.\\
Let be
$
x_1=X_+(\sigma_1),
x_2=X_+(\sigma_2),
$indeed if $\sigma_1<\sigma_2$ by strictly monotonicity of $H'$ we can say that $x_1\neq x_2$. We obtain hence
\begin{align*}
\ell^\star &= x_1 +o(1),\\
&= x_2 +o(1),
\end{align*}
which is absurd.
\item  $\sigma_\star<\underline{\sigma}<\sigma^\star$ and $\overline{\sigma}=\sigma^\star$. In this case we are not allowed Gronwall inequality as in Proposition \ref{effetto fisa}, but under these assumptions we know that $\dot{m}_+\geqslant 0$, hence there exist
\begin{equation}
\label{limite m+}
\lim_{t\to\infty}m_+(t)=m_+^\infty.
\end{equation}
  We consider some $\sigma$ close to the extremal value $\sigma^\star$, $\sigma <\sigma^\star$,
we are going to perform our analysis on the conserved quantity
 \begin{equation}\label{conservato}
 \ell^\star=m_-(t)X_-\left(\sigma(t)\right)+m_0(t)X_0\left(\sigma(t)\right)+ m_+^\infty X_+(\sigma(t))+o(1),
 \end{equation}
  as long as $\sigma\nearrow\sigma^\star$. We will be forced to divide the proof in sub-cases.
 
\begin{enumerate}
\item Suppose $H'''(x_\star)<0$, whence in this case the asymptotic performed in Lemma \ref{asymptotic expansion of the invertible branches} holds.\\ 
 Inserting \eqref{stimesigma} into \eqref{conservato}, we obtain 
 \begin{equation}
 \label{ausiliaria4}
\ell^\star= m_+^\infty X_+(\sigma^\star)+x_\star\left(m_-(t)+m_0(t)\right)
+c\left(m_0(t)-m_-(t)\right)\left(\Delta\sigma\right)^{1/2}
+\mathcal{O}\left(\Delta\sigma\right)
+o(1).
  \end{equation}
  We remark the fact that, thanks to \eqref{espansione Azero/meno} the term $\mathcal{O}\left(\Delta\sigma\right)$ is $ \mathcal{C}^2_\text{loc}\left(\mathbb{R}\right) $.\\
Now, $\ell^\star$ a constant it has to be independent from $\left(\Delta\sigma\right)^{1/2}$, to this end evaluate
\begin{equation}
\label{uffa}
0=\frac{\partial \ell^\star}{\partial\left(\left(\Delta\sigma\right)^{1/2}\right)}=
c\left(m_0(t)-m_-(t)\right)
+\mathcal{O}\left(\left(\Delta\sigma\right)^{1/2}\right).
\end{equation}
From the equation above whence we obtain that
\begin{equation}
\label{approssimazione strana}
m_0(t)=m_-(t)+\mathcal{O}\left(\left(\Delta\sigma\right)^{1/2}\right).
\end{equation}
Moreover
\begin{equation}
\label{m0+m-}
m_0(t)+m_-(t)=1-m_+(t)=1-m_+^\infty+o(1)=B+o(1).
\end{equation}
Considering the result above with the one in equation \eqref{approssimazione strana} we can argue that
\begin{equation}
\label{correlazione m-}
m_-(t)=\frac{B+o(1)+\mathcal{O}\left(\left(\Delta\sigma\right)^{1/2}\right)}{2}=m_-^\infty+o(1)+\mathcal{O}\left(\left(\Delta\sigma\right)^{1/2}\right).
\end{equation}
Moreover a similar approximation is valid also for $m_0$ thanks to \eqref{approssimazione strana}, whence
\begin{equation}
\label{correlazione m0}
m_0(t)=m_0^\infty+o(1)+\mathcal{O}\left(\left(\Delta\sigma\right)^{1/2}\right).
\end{equation}

This approximation is valid as long as the approximation \eqref{ausiliaria4} holds.
Thanks to \eqref{approssimazione strana} we can say that $m_-^\infty=m_0^\infty$, whence, if $\sigma$ oscillates, close to the bifurcation point the masses stabilize around the same value.\\

  Consider now the term $c\left(m_0(t)-m_-(t)\right)$ appearing in \eqref{ausiliaria4}. We know, thanks to Remark \ref{ultimo rem} that
  \begin{equation}
  \label{roba de giustificar}
  c\left(m_0(t)-m_-(t)\right)=-2\left(a_2^{-/0}\left(1-m_+^\infty\right)-m_+^\infty b\right)\left(\Delta\sigma\right)^{1/2}+\mathcal{O}
  \bigl(\left(\Delta\sigma\right)\bigr),
    \end{equation}
  with $c=\left| a_1^{-/0}\right|>0$. Notice that thanks to \eqref{coeff 2} and the fact that $H^{(4)}\left( x_\star\right)>0$ results that $-2\left(a_2^{-/0}\left(1-m_+^\infty\right)-m_+^\infty b\right)>0$.\\
  At this point, consider an interval $[ t_1, t_2 ]$ such that $\sigma<\sigma^\star$ is decreasing in such interval. This interval always exists since $\sigma\in\mathcal{C}^1$ and $\underline{\sigma}<\sigma^\star$, and such that $\sigma(t_2)$ is close enough such that the approximation \eqref{roba de giustificar} still holds. Thanks to \eqref{roba de giustificar} we can argue that in $[t_1,t_2]$, $m_0-m_-$ is an increasing function. On the other hand as long as $\sigma<\sigma^\star$ we have that $\dot{m}_0\leqslant 0$ and $\dot{m}_-\geqslant 0$, hence $m_0-m_-$ has to be decreasing according to this consideration. This is indeed an absurd, proving that $\sigma$ can not oscillate.
   
\item Suppose at last that $H'''(x_\star)=0$, whence, with the same considerations which have been done before
$$
\sigma-\sigma^\star=\frac{H^{(4)}(x_\star)}{6}A_{-/0}^3(\sigma)+\mathcal{O}\left(A_{-/0}^4(\sigma)\right).
$$
Recall that we are considering $H^{(4)}(x_\star)> 0$. Set $\Delta\sigma$ as in the point a. With a procedure similar to the one performed in Remark \ref{fz implicita} we can conclude that
$
A_{-/0}(\sigma)=c\left(\Delta\sigma\right)^{1/3}+o\left(\left(\Delta\sigma\right)^{1/3}\right).
$\\

 $c$ in the equation above that takes the following explicit expression
$
c=\left(-\frac{H^{(4)}\left(x_\star\right)}{6}\right)^{-1/3}<0,
$
whence, for $\sigma<\sigma^\star$ we have that
$
X_{-/0}\left(\sigma\right)-x_\star=c\left(\Delta\sigma\right)^{1/3}+o\left(\left(\Delta\sigma\right)^{1/3}\right)<0,
$
This indeed implies that $X_{-/0}\left(\sigma\right)<x_\star$ for $\sigma$ sufficiently close to $\sigma^\star$, but this contradicts the definition of $X_0$, hence we obtained an absurd.

\end{enumerate}

\end{enumerate}
\end{proof}

\begin{rem}\label{ultimo rem}
We want to justify equation \eqref{roba de giustificar}. Define $R(\sigma)$ the $\mathcal{O}\left(\left(\Delta\sigma\right)\right)$ function appearing in \eqref{ausiliaria4}. Whence
$$
\frac{\partial R}{\partial\left(\left(\Delta\sigma\right)^{1/2}\right)}=\mathcal{O}\left(\left(\Delta\sigma\right)^
{1/2}\right)
$$ 
where $\mathcal{O}\left(\left(\Delta\sigma\right)^
{1/2}\right)$ is the function appearing in \eqref{uffa}.  Thanks to \eqref{espansione Azero/meno} and \eqref{correzione 3} we can express $R(\sigma)$ , i.e.
$$
R(\sigma)=m_0(t)a_2^0\Delta\sigma+m_-(t)a_2^-\Delta\sigma-m_+^\infty b\Delta\sigma+o(1)+\mathcal{O}\left(\left(\Delta\sigma\right)^
{3/2}\right).
$$
Recall that, thanks to \eqref{coeff 2} $a_2^0=a_2^-=a$, and that the term $\mathcal{O}\left(\left(\Delta\sigma\right)^
{3/2}\right)$ is $\mathcal{C}^1$ around 0 and that $o(1)$ depends only on $t$. Using \eqref{m0+m-} we obtain that
$$
R(\sigma)=\left( a\left(1-m_+^\infty\right)-m_+^\infty b\right)\Delta\sigma+o(1)+\mathcal{O}\left(\left(\Delta\sigma\right)^
{3/2}\right).
$$
At this point differentiate both sides of the equation above obtaining 
$$
\frac{\partial }{\partial\left(\left(\Delta\sigma\right)^{1/2}\right)}R(\sigma)=
2\left( a\left(1-m_+^\infty\right)-m_+^\infty b\right)\left(\Delta\sigma\right)^{1/2}+\mathcal{O}\left(\Delta\sigma
\right),
$$
which proves \eqref{roba de giustificar}.\hspace*{\fill}$\blacklozenge$\medskip
\end{rem}
\appendix
\section{Estimates for the fixed point theorem and other technicalities.}

\begin{lemma}\label{invertibilità L} Consider the operator $\mathcal{L}$ as defined in \eqref{definition L}, than $\mathcal{L}$ is invertible.

\end{lemma}
\begin{proof}
Indeed equation (\ref{definition L}) can be seen as
\begin{equation}
\mathcal{L}\phi(\xi)=\phi(\xi)-mK_0\star \phi (\xi)-(1-m)K_+\star \phi (\xi),
\label{L as convolution}
\end{equation}
for the following convolution kernels
\begin{eqnarray*}
K_0(\xi)= a e^{-a\xi}\chi_{[0,+\infty)}(\xi),\\
K_+(\xi)=b e^{-b\xi}\chi_{[0,+\infty)}(\xi).
\end{eqnarray*}

We want to check the invertibility of $\mathcal{L}$ in (\ref{definition L}).\\
To begin our analysis consider first the regularity of the left hand side in (\ref{starting equation}). By definition we have that $\phi(\xi)=\sigma(\xi)-\sigma_0$ is $\mathcal{C}^3$, in fact inherits the same regularity of $H$ around $x_0$, in some interval of the form $(-\infty,T]$ for some $T\in\mathbb{R}$. Hence by the convolution structure of the equation (\ref{L as convolution}) we have that the left hand side, and hence the right hand side of (\ref{starting equation}) is $\mathcal{C}^3_\text{loc}$.\\
Since (\ref{L as convolution}) presents convolutions it seems reasonable to perform some change of variable in such a way to have our functions defined on the positive real line, hence apply the Laplace transform to obtain some information.\\
Setting $z=\xi-\kappa$ for $\kappa\geq 0$ equation (\ref{definition L}) reads as
$$
\phi(\xi)-ma\int_0^\infty e^{-a\kappa}\phi(\xi-\kappa)\d\kappa-(1-m)b\int_0^\infty e^{-b\kappa}\phi(\xi-\kappa)\d\kappa,
$$
which is again a convolution equation, and, moreover the convolution kernels didn't change structure. We are performing an asymptotic analysis for $\xi$ close to $-\infty$, hence it seems reasonable, at least at the moment, to consider $\xi$ bounded from above by some value $\xi_0$. Hence we can write $\xi=\xi_0-x$ for $x\geq 0$, setting
$
\phi(\xi)=\phi({\xi}_0-x)=\psi(x),
$
we can read the above equation as
 $
\psi(x)-ma\int_0^\infty e^{-a\kappa}\psi(x-\kappa)\d\kappa-(1-m)b\int_0^\infty e^{-b\kappa}\psi(x-\kappa)\d\kappa,
$
i.e.
$
\mathcal{L}\psi(x)=\psi(x)-mK_0\star \psi (x)-(1-m)K_+\star \psi (x),
$
which is an equation in convolution form defined on the positive real line.

Performing the substitution
$
W(\xi)=W({\xi}_0-x)=V(x),
$
equation (\ref{starting equation}) turn into
\begin{equation}
\mathcal{L}\psi(x)=V(x),
\label{transformed equation}
\end{equation}
which has the same regularity of (\ref{starting equation}) but is defined on positive numbers.\\
We can hence apply Laplace transform on both sides of (\ref{transformed equation}) obtaining
$
L\psi(\theta)(1-mLK_0(\theta)-(1-m)LK_+(\theta))=LV(\theta),
$
were $L\psi, LK_0, LK_+$ have respectively the domain
\begin{align*}
D(L\psi)&=\{\theta: \mbox{Re} \theta >a\},\\
D(LK_0)&=\{\theta: \mbox{Re} \theta >-a\},\\
D(L\psi)&=\{\theta: \mbox{Re} \theta >-b\}.
\end{align*}
 Hence we can express $L\psi (\theta)=LV(\theta)/C(\theta)$ as a meromorphic function defined on the half complex line $ D\left(LK_0\right)=\left\lbrace\theta: \mbox{Re} \theta >-a \right\rbrace
 $. Our aim is to invert the term on the right hand side of this previous equation via inverse Laplace transform.
 
In particular we want to show that we can express for this particular case the inverse Laplace transform as a residual evaluation, to do so first we have to prove that $|L\psi(\theta)<M/|\theta|^c$ for some $c>0$, hence
$$
|L\psi(\theta)|=\left|\int_0^\infty e^{-\theta x}\psi(x)dx\right|\\
=\left|\int_0^\infty\left[-\frac{1}{\theta}\frac{d}{dx}e^{-\theta x}\right]\psi(x)dx\right|\overset{\text{IbP}}{\leqslant}\left|\left. -\frac{1}{\theta}e^{-\theta x}\psi(x)\right|_0^\infty\right|+\frac{1}{\theta}\left|\int_0^\infty e^{-\theta x}\psi'(x)dx\right|\leqslant\frac{c}{\theta}.
$$
Where in the last inequality we have proceed as follows. Consider
$
\frac{\d}{\d x}\psi (x)=-\frac{\d}{\d \xi}\phi(\xi)=-\frac{\d}{\d \xi}\sigma(\xi),
$ 
and, since $\sigma(\xi)=\int H'(x)\partial_x R(x,\xi)\d x$ with $R$ defined in \eqref{R per contrazione} we obtain
$$
\frac{\d}{\d \xi}\sigma(\xi)=H''(X_+(\xi))\left(\sigma(\xi)-H'(X_+)\right)\\
=\left[H''(x_+)+\mathcal{O}\left(X_+(\xi)-x_+\right)\right]\mathcal{O}\left(X_+(\xi)-x_+\right)<c<\infty.
$$ 
 Now, we know that $\psi$ is continuous, we have to check were are localized the poles of $L\psi$ i.e. the zeroes of $C$ to understand if we can indeed invert the operator $\mathcal{L}$.\\
 Set the equation
 $$
 C(\theta)=1-\frac{ma}{a+\theta}-\frac{(1-m)b}{b+\theta}=0.
$$
After some algebra we obtain two roots
$
\begin{array}{cc}
\theta_1 =0, & \theta_2=-(1-m)a-mb<0,
\end{array}
$
which have real part strictly smaller than $-a$, in this way we obtain the following expression for $\psi$

$$
\psi(x)=2\pi\sum_{i=1,2}\res_{\theta=\theta_i}\left(e^{\theta x}\frac{LV(\theta)}{C(\theta)} \right),
$$
where this last equation is justified by \cite[Theorem 16.39]{apostol}.\\

\end{proof}

\begin{lemma}\label{funzione di green per operatore}
The equation
$
\mathcal{L}G(\xi)=\delta_0(\xi-a),
$
is solved for the function $G(\xi-a)$ where
\begin{equation}
\label{funzione Green}
G(x)=\left[\delta_0(x)+G_r(x)\right]\chi_{\mathbb{R}_+},
\end{equation}
with $G_r(x)=\left[c_1+c_2 e^{\theta_2x} \right]\chi_{\mathbb{R}_+}(x)$, $c_1\neq 0,\theta_2<0$.
\end{lemma}
\begin{proof}
At first I want to show that every solution of
\begin{equation}
\left\lbrace
\begin{array}{lcr}
\mathcal{L}\phi (\xi)=0,&\mbox{for}&\xi\leq\xi_0,\\
\displaystyle\lim_{\xi\to -\infty}\phi(\xi)=0,
\end{array}
\right.
\label{equazione quasi omogenea}
\end{equation}
satisfies $\phi(\xi)=0$ for $\xi\leq\xi_0$.\\
Performing the usual substitution $\xi=\xi_0-x$ we obtain that (\ref{equazione quasi omogenea}) is equivalent to
$$
\left\lbrace
\begin{array}{lcr}
\mathcal{L}\psi (x)=0,&\mbox{for}&x\geq 0,\\
\displaystyle\lim_{x\to \infty}\psi(x)=0.
\end{array}
\right.
$$
Applying Laplace transform we obtain that $L\psi(\theta)C(\theta)=0$ if $D=\{\mbox{Re}\theta>-a\}$, but $C$ is different form zero in $D$, which implies that $L\psi$ have to be identically zero in $D$. This means that
$
L\psi(\theta)=\int_0^\infty e^{-\theta x}\psi(x)dx=0,
$
in $D$, now we can always see Laplace transform of a function $f$ as Fourier transform of an associated function, up to a constant i.e.
$$
Lf(\theta)=\int_0^\infty e^{-\theta x}f(x)dx=\int_{-\infty}^{+\infty}e^{-i\mbox{Im}\theta x}\left[\chi_{\mathbb{R}_+}(x)e^{-\mbox{Re}\theta x}f(x)\right]dx=\mathcal{F}(f_\theta)(\mbox{Im}\theta),
$$
this is a more flexible way to face the problem. Is easy to check that $\chi_{\mathbf{R}_+}(\bullet) e^{-\mbox{Re}\theta   \bullet}\psi(\bullet)=\psi_\theta\in L^2(\mathbb{R})$, hence as long as $\theta\in D$ we obtain that $\mathcal{F}(\psi_\theta)=0$ in $L^2$. Fourier transform is an invertible operator in $L^2$, hence $\psi_\theta=0$ in $L^2$, but since $\psi$ is continuous we obtain that $\psi_\theta$ is identically zero and hence $\psi(x)=0$ for $x\geq0$.\\
Consider now the equation
\begin{equation}
\label{equazione delta}
\mathcal{L}\phi(\xi)=\delta_a(\xi).
\end{equation}

Since $\psi=0$ for $\xi<a$, with the substitutions 
\begin{eqnarray*}
\xi-a=x,\\
\zeta +a = z,\\
\phi(\xi)=\varphi(\xi-a)=\varphi(x).
\end{eqnarray*}
Equation (\ref{equazione delta}) turns into the following
\begin{equation}
\varphi(x)-ma\int_0^xe^{-a(x-\zeta)}\varphi(\zeta)d\zeta-(1-m)b\int_0^xe^{-b(x-\zeta)}\varphi(\zeta)d\zeta=\delta_0(x).
\label{equazione delta formulazione esplicita}
\end{equation}

A straightforward application of transform methods may not be efficient, given the strong discontinuity presented in the problem given by the function $\delta_0$, in this spirit we try to substitute $\varphi$ with a suitable decomposition that may lead to a problem sufficiently regular to apply the Laplace transform.\\
To do so, consider $\varphi$ as $\varphi=\delta_0+G_r$. Inserting this function in such a form equation (\ref{equazione delta formulazione esplicita}) reads as:
$$
\mathcal{L}G_r(x)=mae^{-ax}+(1-m)be^{-bx},
$$
where the member on the right hand side is suitable for application of transform methods. Applying Laplace transform to both sides we obtain the following equation defined in the half complex plane $\{\mbox{Re}\theta>-a\}$
$$
\left(1-\frac{ma}{\theta+a}-\frac{(1-m)b}{\theta+b} \right)LG_r(\theta)=\frac{ma}{\theta+a}+\frac{(1-m)b}{\theta+b},
$$
which is equivalent, after some algebraic manipulation to the following
$$
LG_r(\theta)=\frac{\theta (ma+(1-m)b)+ab}{\theta^2+((1-m)a+mb)\theta}=\frac{B(\theta)}{C(\theta)},
$$
We have hence obtained that $LG_r$ can be expressed as a meromorphic function which has 2 simple poles located at $\theta_1=0,\theta_2<0$. Heaviside inversion theorem (see, for instance, \cite{heaviside} for a proof of such), can be applied, expressing $G_r$ as
$$
G_r(x)=\frac{B(\theta_1)}{C'(\theta_1)}e^{\theta_1 x}+\frac{B(\theta_2)}{C'(\theta_2)}e^{\theta_2 x}=c_1+c_2 e^{\theta_2x} .
$$

We've obtained that if $G(x)$ is solution to (\ref{equazione delta}) we have that $G(x)=\left[\delta_0(x)+G_r(x)\right]\chi_{\mathbb{R}_+}$ with $G_r(x)=0$ as long as $x<0$ and $G_r\in L^\infty(\mathbb{R})$ since $\theta_2<0$.\\
\end{proof}

 In the proof of Proposition \ref{well posedness at -infty} we use a fixed point argument, i.e. we want to show that the operator $T$ defined by mean of formula \eqref{operatorT} is indeed a contraction between Banach spaces. To do so we require the following estimates.
 \begin{prop}\label{prop estimate W}
 Let be $W$ be defined by means of \eqref{W}, than we have that there exists some positive $C=C(a,b,m,\delta)<\infty$ uniformly in $\delta$, and $\xi<\xi_0$ sufficiently negative such that
 \begin{equation}\label{estimate W}
|W(\xi)|\leqslant C e^{-2a\xi}
\end{equation}
 \end{prop}
 \begin{proof}
 Recall that $W$ is defined as follows
\begin{multline}\tag{\ref{W}}
W(\xi)=m\left( -a\mathcal{U}_0[\rho_0](\xi,\sigma)+\rho_0(\xi,\sigma)\right)+(1-m)\left(-b\mathcal{U}_+[\rho_+](\xi,\sigma)+\rho_+(\xi,\sigma)\right)+\\
+ma\int_\mathbb{R}\left(\mathcal{U}_0[\rho_0](\xi,\sigma)-\mathcal{U}_0[\rho](\xi,K,\sigma) \right)Q'(K)dK+m\int_\mathbb{R}\left(\rho(\xi,K,\sigma)-\rho_0(\xi,\sigma) \right)Q'(K)dK
\end{multline}
Recall as well that
\begin{equation}
\label{rho's}
\begin{array}{c}
\rho_0 =\mathcal{O}\left( Y_0^2\right)\\
\rho_+=\mathcal{O}\left( Y_+^2\right)\\
\rho=\mathcal{O}\left( \left(Ke^{-a\xi}+Y\right)^2\right)
\end{array}
\end{equation}
and $Y_0,Y_+,Y=o(1) $ as $\xi\to -\infty$. \\
Recall as well, as shown in Lemma \ref{dimostrazione condizione H}  that
$
\left|Y_{0/+}(\xi)\right|\leqslant e^{(-a+\delta)\xi},
$
hence
$
\left|\rho_{0/+}(\xi)\right|\leqslant C_H e^{2(-a+\delta)\xi}.
$
A straightforward computation shows moreover that
$
\left|\mathcal{U}_0[\rho_{0/+}](\xi,\sigma)\right|\leqslant \tilde{C}_H e^{2(-a+\delta)\xi}.
$
By the definition of $\rho(K,\xi)$ we know that
$
\rho(K,\xi)=\mathcal{O}\left(\left(X(K,\xi)-x_0\right)^2\right),
$
and, since 
$
\lim_{K\to \pm \infty}\left|X(K,\xi)\right|=\left|X_\pm(\xi)\right|<\infty,
$
we can say that the function $X(\cdot,\xi)$ is bounded as long as the extremal characteristics $X_\pm$ exist. In particular since $X_\pm =x_\pm +Y_\pm$ with $Y_\pm(\xi)\leqslant e^{(-a+\delta)\xi}$ we can say that
\begin{align*}
X_+(\xi)&\leqslant x_++\frac{1}{2},\\
X_-(\xi)&\geqslant x_--\frac{1}{2},
\end{align*}
hence we can bound uniformly $|X(K,\xi)|\leqslant x_+-x_-+1$, which implies that
$
\left|\frac{1}{\sqrt{\pi}}\int_\mathbb{R}\rho(K,\xi)e^{-K^2}\d K\right|\leqslant \bar{C}_H e^{2(-a+\delta)\xi}.
$
We point out the fact that $Q'(K)=\frac{1}{\sqrt{\pi}}e^{-K^2}$.\\
A straightforward computation as before shows hence that
$
\left|\frac{1}{\sqrt{\pi}}\int_\mathbb{R}\mathcal{U}_0\left[\rho\right](K,\xi)e^{-K^2}\d K\right|\leqslant \bar{\bar{C}}_H e^{2(-a+\delta)\xi}.
$
At this point, considering $W$ as in equation \eqref{W} we can hence assert that
\begin{equation}\label{estimate W}
|W(\xi)|\leqslant C e^{2(-a+\delta)\xi}
\end{equation}
for some positive $C=C\left(a,b,m,\delta,H\right)<\infty$ uniformly in $\delta$ and $\xi<\xi_0$ sufficiently negative.\\
 \end{proof}
 
 \begin{prop}\label{prop deltaW}
 Let $W$ be defined via equation \eqref{W}, and let $\phi_1,\phi_2\in K_{M,\xi_0}(\delta)$. Consider $H$ with the same properties as in Lemma \ref{dimostrazione condizione H}. We recall again that all the functions $\rho_i,Y_i,i=0,+$ have a direct dependence on the parameter $\sigma$, which is equivalent as having a dependence for a parameter $\phi\in K_{M,\xi_0}(\delta)$, than
 \begin{equation}\label{deltaW}
\left|W(\xi,\phi_1)-W(\xi,\phi_2)\right|\lesssim e^{-(2a+\delta)\xi}\cdot \left\|\phi_1-\phi_2\right\|_{K_{M,\xi_0}(\delta)}.
\end{equation}
\end{prop} 
 \begin{proof}
The functions $\rho_i,Y_i,i=0,+$ have a direct dependence on $\phi\in K_{M,\xi_0}(\delta)$, in particular the dependence is given by
\begin{align}
Y_{0/+}&=\mathcal{U}_{0/+}\left[\phi
\right]-\mathcal{U}_{0/+}\left[\rho_{0/+}\right],\label{Y_0}\\
Y&=\mathcal{U}_{0}\left[\phi
\right]-\mathcal{U}_{0}\left[\rho\right],
\end{align}
and the functions $\rho_i$ are defined as in \eqref{rho's}.\\
We want to derive first some estimate for the element $\left|\mathcal{U}_0\left(\phi_1-\phi_2\right)(\xi)\right|$, with $\mathcal{U}_0\left[f\right](\xi)=\int_{-\infty}^{\xi} e^{-a(\xi-s)}f(s)\d s,$ hence
\begin{multline}\label{stimaphi}
\left|\mathcal{U}_0\left(\phi_1-\phi_2\right)(\xi)\right|=\left|\int_{-\infty}^{\xi} e^{-a(\xi-z)}\left(\phi_1(z)-\phi_2(z)\right)\d z\right|\\
\leqslant\left\|\phi_1-\phi_2\right\|_{K_{M,\xi_0}(\delta)}\cdot e^{-a\xi}\int_{-\infty}^{\xi} e^{-(a+\delta)z}\d z\leqslant\frac{1}{-(a+\delta)} e^{-(2a+\delta)\xi}\cdot\left\|\phi_1-\phi_2\right\|_{K_{M,\xi_0}(\delta)},
\end{multline}
this estimate will be useful in the following.\\
Since $H$ satisfies the same properties as the $H$ considered in Lemma \ref{dimostrazione condizione H} we can say that, for $\xi_0$ sufficiently negative,
$
\rho_0\left(\xi,\phi_i\right)=R\left(Y_0(\xi,\phi_i)\right) ,
$
where we decide to express $R\left(Y_0(\xi,\phi_i)\right)$ in its integral form, i.e.
$$
R\left(Y_0(\xi,\phi_i)\right)=\int_0^{Y_0(\xi,\phi_i)}\frac{H'''\left(x_0+\zeta\right)}{2}\left(Y_0(\xi,\phi_i)-\zeta\right)^2\d \zeta,
$$
since it gives a better idea of the regularity of the reminder. We can express the reminder in such a way since by our hypothesis in Proposition \ref{well posedness at -infty} $H'''$ is locally absolutely continuous. Since $Y_0$ is at least $\mathcal{C}^1$ (see Lemma \ref{dimostrazione condizione H}), then $R(Y_0)$ is $\mathcal{C}^1$ as well.\\
Consider at this point the function
$
R(x)=\int_0^x\frac{H'''(x_0+\zeta)}{2}\left(x-\zeta\right)^2\d \zeta.
$
$R\in\mathcal{C}^{0,1}_{\text{loc}}\cup \mathcal{C}^{1}_{\text{loc}}$, hence for $x,y$ in a compact set $\mathcal{K}$
$|R(x)-R(y)|\leqslant L(\mathcal{K})|x-y|.
$ We claim that
\begin{equation}
\label{lipschitz local}
\lim_{\mathcal{L}(\mathcal{K})\to 0, 0\in\mathcal{K}}L(\mathcal{K})=0.
\end{equation}
where $\mathcal{L}$ is the Lebesgue measure on $\mathbb{R}$.\\
Indeed, for locally differentiable functions $g$ on compact sets $\mathcal{K}$, we have that
$
\left\|g\right\|_{\mathcal{C}^{0,1}_{\mathcal{K}}}\leqslant \displaystyle\sup_{x\in\mathcal{K}} |g'(x)|,
$
which means we can bound the Lipschitz norm in term of the sup of the derivative in the compact set. This means that, if we prove that $R'(x)\to 0$ as $|x|\to 0$ we prove \eqref{lipschitz local}.\\
A computation shows that
$
R'(x)=x\int_0^x H'''(x_0+z)\d z-\int_0^x zH'''(x_0+z)\d z \xrightarrow{x\to 0}0,
$
since $H\in\mathcal{C}^3_{\text{loc}}$.\\
Now, $R\left( Y_0(\xi,\phi_i)\right)=\rho_0(\xi,\phi_i)$, this means we have obtained that
\begin{equation}\label{procedura}
\left|\rho_0\left(\xi,\phi_1\right)-\rho_0\left(\xi,\phi_2\right)\right|\leqslant L(\xi)\cdot\left|Y_0\left(\xi,\phi_1\right)-Y_0\left(\xi,\phi_2\right)\right|,
\end{equation}
with $L(\xi)=o(1)$ as $\xi\to -\infty$.

Moreover, since $Y_0$ is defined via equation \eqref{Y_0}, we can infer that
$$
\left|Y_0\left(\xi,\phi_1\right)-Y_0\left(\xi,\phi_2\right)\right|\leqslant \left|\mathcal{U}_0\left(\phi_1-\phi_2\right)(\xi)\right|+\left|\mathcal{U}_0\left(\rho_0\left(\phi_1\right)-\rho_0\left(\phi_2\right)\right)(\xi)\right|,
$$
obtaining,
$$
\left|\rho_0\left(\xi,\phi_1\right)-\rho_0\left(\xi,\phi_2\right)\right|\leqslant L(\xi)\Bigl(\left|\mathcal{U}_0\left(\phi_1-\phi_2\right)(\xi)\right|+\left|\mathcal{U}_0\left(\rho_0\left(\phi_1\right)-\rho_0\left(\phi_2\right)\right)(\xi)\right|\Bigr).
$$
Hence, considering the term $-a\mathcal{U}_0[\rho_0(\phi)](\xi)+\rho_0(\xi,\phi)$ in \eqref{W}, with these estimates, we can argue that
\begin{multline}\label{Schifezza}
\left| -a\left( \mathcal{U}_0\left(\rho_0\left(\phi_1\right)-\rho_0\left(\phi_2\right)\right)(\xi)\right)+\rho_0\left(\xi,\phi_1\right)-\rho_0\left(\xi,\phi_2\right)\right|\\
\leqslant L(\xi) \left|\mathcal{U}_0\left(\phi_1-\phi_2\right)(\xi)\right|+\left(-a+L(\xi)\right)\left|\mathcal{U}_0\left(\rho_0\left(\phi_1\right)-\rho_0\left(\phi_2\right)\right)(\xi)\right|,
\end{multline}
hence what's left to understand is how to bound the term $\left|\mathcal{U}_0\left(\rho_0\left(\phi_1\right)-\rho_0\left(\phi_2\right)\right)(\xi)\right|$.\\
Repeating the procedure in equation \eqref{procedura} we can argue that, there exists a $0<q<1$ that we can make as small as we want since 
$
q=\displaystyle\sup_{\xi\leqslant \xi_0}\left\lbrace L(\xi)\right\rbrace,
$ 
such that,
\begin{multline*}
\left|\mathcal{U}_0\left(\rho_0\left(\phi_1\right)-\rho_0\left(\phi_2\right)\right)(\xi)\right|\leqslant q\left|\mathcal{U}_0\left(Y_0(\xi,\phi_1)-Y_0(\xi,\phi_2)\right)\right|\\
\overset{\eqref{Y_0}}{=} q\left| \mathcal{U}_0^2\left(\phi_1-\phi_2\right)(\xi)-\mathcal{U}^2_0\left(\rho_0\left(\phi_1\right)-\rho_0\left(\phi_2\right)\right)(\xi)\right|\\
\leqslant\sum_{n=1}^N q^n \left| \mathcal{U}_0^{n+1}\left(\phi_1-\phi_2\right)(\xi)\right|+
q^N \left|\mathcal{U}^{N+1}_0\left(\rho_0\left(\phi_1\right)-\rho_0\left(\phi_2\right)\right)(\xi)\right|.
\end{multline*}
Is easy to check that
$$
\left|\mathcal{U}^{N+1}_0\left(\rho_0\left(\phi_1\right)-\rho_0\left(\phi_2\right)\right)(\xi)\right|\leqslant \left(\frac{1}{-a+2\delta}\right)^{N+1} e^{2(-a+\delta)\xi},
$$
which in turn implies the following
$$
q^N \left|\mathcal{U}^{N+1}_0\left(\rho_0\left(\phi_1\right)-\rho_0\left(\phi_2\right)\right)(\xi)\right|\leqslant\frac{1}{q} \left(\frac{q}{-a+2\delta}\right)^{N+1} e^{2(-a+\delta)\xi}\xrightarrow{N\to\infty}0,
$$
if $q<-a+2\delta$. At this point we can hence say that
\begin{equation}
\label{deltarhozero}
\left|\mathcal{U}_0\left(\rho_0\left(\phi_1\right)-\rho_0\left(\phi_2\right)\right)(\xi)\right|\leqslant \sum_{n=1}^\infty q^n \left| \mathcal{U}_0^{n+1}\left(\phi_1-\phi_2\right)(\xi)\right|,
\end{equation}
which means that we have bounded the term $\left|\mathcal{U}_0\left(\rho_0\left(\phi_1\right)-\rho_0\left(\phi_2\right)\right)(\xi)\right|$ from above with another term which we can evaluate thanks to equation \eqref{stimaphi}.\\
With the estimate in \eqref{stimaphi} we can bound from above the term on the right hand side of \eqref{deltarhozero} with
$$
\sum_{n=1}^\infty q^n \left| \mathcal{U}_0^{n+1}\left(\phi_1-\phi_2\right)(\xi)\right|\leqslant \sum_{n=1}^\infty \frac{1}{q}\left( \frac{q}{-(a+\delta)}\right)^{n+1} e^{-(2a+\delta)\xi}\cdot\|\phi_1-\phi_2\|_{K_{M,\xi_0}(\delta)},
$$
which, considering \eqref{deltarhozero} gives us 
\begin{equation}
\label{estimates rho zero final}
\left|\mathcal{U}_0\left(\rho_0\left(\phi_1\right)-\rho_0\left(\phi_2\right)\right)(\xi)\right|\leqslant C\left(q,a,\delta\right)e^{-(2a+\delta)\xi}\cdot\left\|\phi_1-\phi_2\right\|_{K_{M,\xi_0}(\delta)}.
\end{equation}
At this point we can hence plug the estimates in \eqref{estimates rho zero final} and \eqref{stimaphi} into \eqref{Schifezza}, this gives us the following
\begin{multline}\label{bound1}
\left| -a\left( \mathcal{U}_0\left(\rho_0\left(\phi_1\right)-\rho_0\left(\phi_2\right)\right)(\xi)\right)+\rho_0\left(\xi,\phi_1\right)-\rho_0\left(\xi,\phi_2\right)\right|\leqslant\\
k_1(q,a,\delta)e^{(-3a+2\delta)\xi}\|\phi_1-\phi_2\|_{K_{M,\xi_0}(\delta)}+k_2(q,a,\delta)e^{-(2a+\delta)\xi}\|\phi_1-\phi_2\|_{K_{M,\xi_0}(\delta)}\\
\leqslant k(q,a,\delta,q)e^{-(2a+\delta)\xi}\|\phi_1-\phi_2\|_{K_{M,\xi_0}(\delta)}.
\end{multline}
With the same procedure just performed, we can obtain the following bound
\begin{equation}\label{bound2}
\left| -b\left( \mathcal{U}_+\left(\rho_+\left(\phi_1\right)-\rho_+\left(\phi_2\right)\right)(\xi)\right)+\rho_+\left(\xi,\phi_1\right)-\rho_+\left(\xi,\phi_2\right)\right|
 \leqslant k(a,\delta,\tilde{q})e^{-(2a+\delta)\xi}\|\phi_1-\phi_2\|_{K_{M,\xi_0}(\delta)},
\end{equation}
and
\begin{equation}\label{bound3}
\left|\int_\mathbb{R}\left( -a\left( \mathcal{U}_0\left(\rho\left(\phi_1\right)-\rho\left(\phi_2\right)\right)(\xi,K)\right)+\rho\left(\xi,\phi_1,K\right)-\rho\left(\xi,\phi_2,K\right)\right)\d K\right|
\leqslant \tilde{k}(a,\delta,q)e^{-(2a+\delta)\xi}\|\phi_1-\phi_2\|_{K_{M,\xi_0}(\delta)}.
\end{equation}
At this point, considering the bounds \eqref{bound1}--\eqref{bound3} and the structure of $W$ which is given explicitly in \eqref{W} we can assert that
$$
\left|W(\xi,\phi_1)-W(\xi,\phi_2)\right|\lesssim e^{-(2a+\delta)\xi}\cdot \left\|\phi_1-\phi_2\right\|_{K_{M,\xi_0}(\delta)}.
$$
Concluding the estimate.
 \end{proof}

 \end{document}